\documentclass[a4paper,UKenglish,cleveref, autoref, thm-restate]{lipics-v2021}

\nolinenumbers
\hideLIPIcs

\title{Approximation Algorithms for Graphic Multi-Depot TSP and Graphic Multi-Path TSP} 

\usepackage{algorithm}
\usepackage{algpseudocode}
\usepackage{mathtools}
\usepackage{tikz}
\usetikzlibrary{arrows.meta}
\usepackage{tikzit}

\tikzstyle{Normal Node}=[font={\scriptsize}, fill=white, draw=black, shape=circle, inner sep=0.1111em, minimum size=1em, scale=0.8]
\tikzstyle{Depot}=[fill=white, draw=black, shape=rectangle, inner sep=0cm, minimum size=8pt, font={\scriptsize}, scale=0.8]
\tikzstyle{DepotGray}=[fill={rgb,255: red,191; green,191; blue,191}, draw=black, shape=rectangle, inner sep=0cm, minimum size=8pt, font={\scriptsize}, scale=0.8]
\tikzstyle{Invis}=[fill=none, draw=none, shape=circle, inner sep=0cm, minimum width=0pt, minimum height=0pt, font={\scriptsize}, scale=0.8]
\tikzstyle{Gray}=[fill={rgb,255: red,191; green,191; blue,191}, draw=black, shape=circle, scale=0.8]
\tikzstyle{Label}=[font={\footnotesize}, text=black, scale=0.8]
\tikzstyle{LabelScriptsize}=[font={\scriptsize}, text=black, scale=0.8]
\tikzstyle{Center}=[fill=rwth-red, draw=black, shape=circle, scale=0.8]

\tikzstyle{Normal Edge}=[draw=black, fill=none, -, line width=1.1pt]
\tikzstyle{Join}=[-, draw=rwth-blue, densely dotted, line width=1.1pt]
\tikzstyle{Opt}=[-, draw=rwth-red, style=densely dotted, line width=1.1pt]
\tikzstyle{Tree}=[-, draw=rwth-green, style=dashed, line width=1.1pt]
\tikzstyle{Invis}=[-, draw=white, line width=1.1pt]
\tikzstyle{EvenCover}=[-, draw=rwth-purple, dashed, line width=1.1pt]
\tikzstyle{DoubleEdges}=[-, draw=rwth-orange, densely dotted, line width=1.1pt]
\tikzstyle{UnusedEdge}=[draw=black, fill=none, -, line width=1.1pt, draw opacity=0.2]
\tikzstyle{Removable1}=[-, draw=rwth-red, style=densely dotted, line width=1.1pt]
\tikzstyle{PairingArrow}=[<->, line width=1.1pt]

\input{img/main.tikzdefs}
\usepackage{subcaption}

\usetikzlibrary{calc, backgrounds}

\newcommand{\parallelEdge}[4][]{%
  \path ($(#2)!#4!90:(#3)$) edge[#1] ($(#3)!#4!-90:(#2)$);%
}

\usepackage{booktabs,multirow,float}

\usetikzlibrary{calc,arrows,automata,patterns,graphs,shapes,petri,decorations.pathmorphing,decorations.markings,decorations.pathreplacing} 
\usetikzlibrary{decorations.markings,arrows.meta}
\usetikzlibrary{decorations.markings,arrows.meta}
\tikzset{->-/.style={decoration={markings,
  mark=at position .5 with {\arrow{Stealth[scale=0.6]}}},
  postaction={decorate}}}
\tikzset{vertex/.style={draw,circle,inner sep=0pt,minimum size=18pt},>=latex'}
\tikzset{c/.style={draw,circle,inner sep=0pt,minimum size=15pt},>=latex'}
\tikzset{dot/.style={draw,circle,inner sep=0pt,minimum size=3pt, fill=black},>=latex'}

\renewcommand{\phi}{\varphi}

\newcommand{\C}{\mathcal{C}}

\newcommand{\R}{\mathbb{R}}
\renewcommand{\P}{\mathcal{P}}

\DeclareMathOperator{\odd}{odd}

\DeclareMathOperator{\comp}{comp}
\DeclareMathOperator{\uncovered}{uc}

\newcommand{\GraphicTSP}{Graphic TSP\xspace}
\newcommand{\GraphicMultiPath}{Graphic Multi-Path TSP\xspace}

\newcommand{\ProblemNonGraphic}{Multi-Depot TSP\xspace}
\newcommand{\Problem}{Graphic Multi-Depot TSP\xspace}

\newcommand{\ALG}{\text{ALG}}
\newcommand{\OPT}{\text{OPT}}

\newcommand{\disjcup}{\mathbin{\dot{\cup}}}

\newcommand{\evencover}{even cover\xspace}
\newcommand{\evencovers}{even covers\xspace}

\newcommand{\nn}[1]{{#1}_{nn}}
\newcommand{\nd}[1]{{#1}_{nd}}
\newcommand{\dd}[1]{{#1}_{dd}}
\newcommand{\completion}[1]{F({#1})}
\newcommand{\centers}{W}
\newcommand{\double}{d}
\newcommand{\bridges}[1]{\text{Br}({#1})}

\usepackage{xspace}

\usepackage{xcolor}

\definecolor{rwth-blue}{cmyk/RGB/HTML/gray}{1,.50,0,0/0,84,159/00549F/1}
\definecolor{rwth-black}{cmyk/RGB/HTML}{0,0,0,1/0,0,0/000000}%
\definecolor{rwth-magenta}{cmyk/RGB/HTML}{0,1,.25,0/227,0,102/E30066}%
\definecolor{rwth-yellow}{cmyk/RGB/HTML}{0,0,1,0/255,237,0/FFED00}%
\definecolor{rwth-petrol}{cmyk/RGB/HTML}{1,.3,.5,.3/0,97,101/006165}%
\definecolor{rwth-turquoise}{cmyk/RGB/HTML}{1,0,.4,0/0,152,161/0098A1}%
\definecolor{rwth-green}{cmyk/RGB/HTML}{.7,0,1,0/87,171,39/57AB27}
\definecolor{rwth-maygreen}{cmyk/RGB/HTML}{.35,0,1,0/189,205,0/BDCD00}%
\definecolor{rwth-orange}{cmyk/RGB/HTML}{0,.4,1,0/246,168,0/F6A800}%
\definecolor{rwth-red}{cmyk/RGB/HTML}{.15,1,1,0/204,7,30/CC071E}%
\definecolor{rwth-bordeaux}{cmyk/RGB/HTML}{.25,1,.70,.20/161,16,53/A11035}%
\definecolor{rwth-violet}{cmyk/RGB/HTML}{.70,1,.35,.15/97,33,88/612158}%
\definecolor{rwth-purple}{cmyk/RGB/HTML}{.6,.6,0,0/122,111,172/7A6FAC}

\definecolor{darkgreen}{rgb}{0.2,0.8,0.55}

\definecolor{ForestGreen}{rgb}{0.0333,0.4451,0.0333}
\definecolor{DarkRed}{rgb}{0.65,0,0}
\definecolor{Red}{rgb}{1,0,0}

\author{Niklas Dahlmeier}{RWTH Aachen University}{}{}{}
\author{Raoul Luqué}{RWTH Aachen University}{}{}{}
\author{Philipp Pabst}{RWTH Aachen University}{}{}{}
\author{Laura Vargas Koch}{RWTH Aachen University}{}{}{}

\authorrunning{N. Dahlmeier, R. Luqué, P. Pabst and L. Vargas Koch} 

\Copyright{Niklas Dahlmeier, Raoul Luqué, Philipp Pabst and Laura Vargas Koch} 



\category{} 

\relatedversion{} 

\begin{document}

\maketitle

\begin{abstract}
We give improved approximation algorithms for two graphic TSP variants. First, we consider \ProblemNonGraphic on graphic instances. In this variant, we are given an undirected graph with a set of depot vertices, and we are tasked with finding an Eulerian multi-edge set that connects each non-depot vertex to at least one depot. For this problem, Deppert, Kaul and Mnich \cite{deppertKaulMnich_MDMTSP} presented a $3/2$-approximation algorithm with running time exponential in the number of depots. We improve this result to a fully polynomial $3/2$-approximation. Our algorithm is a variant of the classic Christofides-Serdyukov algorithm with an improved and simple analysis. On subcubic instances, we improve the approximation guarantee to $4/3$ using the technique of removable pairings introduced by Mömke and Svensson \cite{mömke_svensson}. We use a removable pairing to construct a set of disjoint cycles that can be completed to a cheap feasible solution.

Second, we consider the Graphic Multi-Path TSP, which was previously considered by Alimi, Dahlmeier, Mömke, Pabst and Vargas Koch \cite{graphic_multipath_paper} who give a 2-approximation for this problem. Additionally, the authors present a reduction that shows that any better-than-two approximation for \Problem immediately leads to an improvement for Graphic Multi-Path TSP. We give an improved algorithm and analysis that yield an approximation guarantee of $1.9304$.
\end{abstract}

\section{Introduction}
The traveling salesperson problem is a fundamental problem in combinatorial optimization. Given an edge-weighted undirected graph $G=(V,E)$, the most basic version of TSP asks for a shortest tour visiting all vertices at least once. We study \emph{Graphic} TSP, the special case of the problem where each edge has a cost of exactly $1$. There is a long line of work \cite{OveisGharanSaberiSingh2011GraphicTSP, mömke_svensson, mucha2012graphicTSP} which shows that this restriction to unit-weight instances allows for better approximation algorithms, with the best result so far being a $1.4$-approximation due to Seb\H{o} and Vygen \cite{seboVygen2014graphic}.

We study two generalizations of \GraphicTSP. The first is \Problem where, in addition to $G = (V,E)$, a set of depots $D \subseteq V$ is given. The goal is to cover all non-depot vertices with a set of closed tours each rooted at a depot. Formally speaking, we want to find an Eulerian multi-edge set of minimum cardinality such that all non-depot vertices $v\in V \setminus D$ are connected to at least one depot $d \in D$. This problem can also be viewed as a vehicle routing problem without capacities, where the objective is to minimize the total path length. 

\newpage

\vspace*{0.3cm}
\noindent\textsc{\Problem}
\begin{description}
  \item[\textbf{Input:}] An undirected, connected graph $G = (V,E)$ with unit edge costs and depots $D \subseteq V$.
  
  \item[\textbf{Find:}] A minimum-size Eulerian multi-edge set such that all non-depot vertices $v \in V\setminus D$ are connected to at least one depot
$d \in D$.
\end{description}

The most well-known algorithm to solve TSP is the Christofides-Serdyukov algorithm \cite{christofides, serdyukov}. This algorithm computes a spanning tree $T$ and then connects the remaining odd degree vertices with an $\text{odd}(T)$-join. Recall that an $\text{odd}(T)$-join is an edge set where exactly the vertices that have odd degree in $T$ have odd degree within this edge set. 
It is easy to see that the cost of the spanning tree is a lower bound on the optimal solution $\OPT$. The cost of the $\odd(T)$-join can be bounded by $\frac{1}{2} \OPT$ as the optimal tour can be shortcut to a tour just visiting the vertices of $\odd(T)$ and taking every second edge yields two disjoint $\odd(T)$-joins. Combined, this yields an approximation guarantee of $3/2$.
Variants of this simple algorithm and analysis have been applied to several TSP settings in the past \cite{an_kleinberg_shmoys_s_t_path_best_of_many_christofides, traub_t_tours_best_of_many_christofides, 2010_1_5_approx_mdmtsp}.

A natural adaptation of the Christofides--Serdyukov algorithm to \Problem is to compute a spanning forest~$F$ in which every component contains a depot, and then to correct parity via an $\odd(F)$-join. However, the cost of this $\odd(F)$-join can no longer be bounded by $\frac{1}{2}\OPT$ since a connected component of an optimal solution may contain an odd number of vertices of $\odd(F)$. Indeed, Xu, Xu and Rodrigues~\cite{XU2011218} showed that this algorithm has a tight approximation ratio of $2-\frac{1}{|D|}$ for \ProblemNonGraphic.

There is a line of literature improving on this guarantee, but at the cost of a running time exponential in~$|D|$.
Xu and Rodrigues~\cite{2010_1_5_approx_mdmtsp} developed a Christofides--Serdyukov variant that achieves a ratio of $\frac{3}{2}$ for a constant number of depots. Deppert, Kaul and Mnich~\cite{deppertKaulMnich_MDMTSP} improved the running time, however, it remains exponential in the number of depots; they obtain a $(\frac{3}{2}+\varepsilon)$-approximation for general metrics and a $\frac{3}{2}$-approximation for \Problem. Very recently, Zhao, Liu and Xiao~\cite{one_plus_square_root_paper} gave the first polynomial-time algorithm that improves on $2-\frac{1}{|D|}$ for an arbitrary number of depots, achieving a ratio of $1+\frac{1}{\sqrt{2}}\approx 1.707$ for general metrics.
\Cref{tab:state_of_the_art} summarizes these results together with their running times.

\begin{table}[tbp]
  \centering
  \footnotesize
  \setlength{\tabcolsep}{3pt}
  \begin{tabular*}{\columnwidth}{@{\extracolsep{\fill}} l l l l @{}}
    \toprule
    Problem & Reference & Ratio & Running time \\
    \midrule
    \multirow{6}{*}{\ProblemNonGraphic}
      & Xu, Rodrigues \cite{2010_1_5_approx_mdmtsp}          & $3/2$                 & $n^{\mathcal{O}(|D|)}$ \\
      & Xu, Xu, Rodrigues\ \cite{XU2011218}                          & $2-1/|D|$               & polynomial \\
      & Traub, Vygen, Zenklusen\ \cite{reducing_path_tsp_to_tsp}        & $\alpha+\varepsilon$  & $n^{\mathcal{O}(|D|/\varepsilon)}$ \\
      & Deppert, Kaul, Mnich\ \cite{deppertKaulMnich_MDMTSP}       & $3/2+\varepsilon$     & $(1/\varepsilon)^{\mathcal{O}(|D|\log|D|)}\, n^{\mathcal{O}(1)}$ \\
      & Zhao, Liu, Xiao\ \cite{one_plus_square_root_paper}       & $1+1/\sqrt{2}$        & polynomial \\
      & Zhao, Su, Xiao\ \cite{zhaoXuSiao2026betterThanThreeHalf} & $\alpha+\varepsilon$ & $2^{\mathcal{O}(|D|/\varepsilon)}\, n^{\mathcal{O}(1)}$ \\
    \midrule
    \multirow{2}{*}{\Problem}
      & Deppert, Kaul, Mnich\ \cite{deppertKaulMnich_MDMTSP}       & $3/2$                 & $2^{|D|}\, n^{\mathcal{O}(1)}$ \\
      & \textbf{This work}                                   & $\mathbf{3/2}$        & \textbf{polynomial} \\
    \bottomrule
  \end{tabular*}
  \caption{Known results; $\alpha < 3/2$ is the best known approximation ratio for TSP.}
  \label{tab:state_of_the_art}
  \vspace{-0.5cm}
\end{table}

As our first main result, we prove that a variant of the Christofides--Serdyukov algorithm is a $\frac{3}{2}$-approximation for \Problem whose running time is polynomial, independent of the number of depots. This removes the exponential dependence on~$|D|$ from the running time of~\cite{deppertKaulMnich_MDMTSP} and, for graphic instances, improves on the ratio of $1+\frac{1}{\sqrt{2}}$ from~\cite{one_plus_square_root_paper}.

\begin{restatable}{theorem}{gmdtspgeneral}
\label{thm:gmdm-tsp-general}
    There is a polynomial time $\frac{3}{2}$-approximation algorithm for the \Problem.
\end{restatable}

\newpage
The idea behind our analysis is that, whenever the cost of the odd-join exceeds $\tfrac{1}{2}\OPT$, our spanning forest is strictly cheaper than an optimal solution, and the saved cost for the spanning forest and the additional cost of the odd-join cancel out.

In more detail, in our algorithm we compute a forest $F$ by computing a spanning tree for every component $C_i$ of $G[V\setminus D]$ and connect it to an arbitrary depot. This differs from previous constructions, where components of $G[V \setminus D]$ might have multiple incident edges to depots in $F$. Afterwards, we compute a minimum-cardinality $\text{odd}(F)$-join as a parity correction to obtain a feasible solution. 
To upper bound the cost of this $\text{odd}(F)$-join, we use the notion of \emph{auxiliary edge sets} due to Xu and Rodrigues \cite{2010_1_5_approx_mdmtsp}. That is, we add edges to the optimal solution to obtain an augmented solution, with the property that every connected component of that solution contains an even number of odd-degree vertices of $F$. We can show that the size of this auxiliary edge set corresponds to the number of edges that $F$ saves compared to the optimal solution.

For Graphic TSP, there are approximation algorithms with better guarantees than Christofides-Serdyukov. The first algorithm for \GraphicTSP that improves on the approximation factor of $3/2$ by a sizable amount is due to Mömke and Svensson \cite{mömke_svensson}. Their main innovation is the introduction of the \emph{removable pairing} framework. Here, the basic idea is to construct a set of \emph{removable edges} $R$ that can be removed from the original graph without violating connectedness. This ensures feasibility when adhering to some pairing constraints. 

A natural question is whether this framework extends to the more general \Problem.
We consider this for graphs with maximum degree of at most $3$, i.e., \emph{subcubic} graphs, as they allow for a simpler method of constructing removable pairings. 
Extending the algorithm and analysis of Graphic TSP in the straightforward way requires a stronger form of $2$-connectedness that also accounts for the connected components of $G[V \setminus D]$. This is because the standard assumption of $2$-connectedness for Graphic TSP does not carry over. 
Thus, we develop a new algorithmic approach while still building on the concept of removable pairings. In contrast to Graphic TSP, we start with a well chosen collection of disjoint cycles that we complement to a feasible solution by ensuring the necessary connectivity.

\begin{restatable}{theorem}{gmdmtspsubcubic}
\label{thm:gmdm-tsp-subcubic}
    There is a polynomial time $\frac{4}{3}$-approximation algorithm for \Problem when restricted to subcubic graphs.
\end{restatable}
 Our algorithm first computes a subset $S$ with $\deg_S(v) \in \{0,2\}$ for all $v\in V$, i.e., a union of disjoint cycles. We call this set an \emph{\evencover} following Wigal, Yoo and Yu \cite{WigalYooYu2021cubic}. 
 The even cover minimizes a weight function whose definition uses the structure of a computed removable pairing. 
 This weight function has two crucial properties. First, a min-cost odd-join for it guarantees the existence of a sufficiently cheap \evencover. Second, it guarantees 
 that we can complement $S$ to a complete, feasible solution $F(S)$ cheaply. 
  To compare $F(S)$ to \OPT, we give an improved lower bound for \OPT.
 The idea of completing even covers to solutions was used previously for Graphic TSP when restricted to cubic (i.e., $3$-regular) and subcubic instances \cite{dvorakKralMohar2017cubic, WigalYooYu2021cubic}.

For subcubic graphs, our guarantee matches the approximation guarantee for Graphic TSP due to Mömke and Svensson \cite{mömke_svensson}.
An exciting follow-up question of our work is whether our approach can be extended to achieve a better-than-$\tfrac{3}{2}$ approximation for \Problem on arbitrary graphs.

As pointed out in \cite{graphic_multipath_paper}, \Problem is a special case of \GraphicMultiPath and the key difficulty in obtaining an approximation factor below 2 for the latter. In \GraphicMultiPath, we are given a collection of source-sink pairs $(s_i, t_i) \in V^2$ for $i \in [k]$. The objective is to find $k$ walks $P_i$ of minimum total length connecting $s_i$ to $t_i$ while ensuring that every vertex is covered by at least one walk $P_i$.
Using \cref{thm:gmdm-tsp-general} we improve on the best known approximation factor for Graphic Multi-Path TSP.

\vspace*{0.3cm}
\noindent\textsc{Graphic Multi-Path TSP}
\begin{description}
  \item[\textbf{Input:}] An undirected, connected graph $G = (V,E)$ with unit edge costs and a set of source-sink pairs $(s_i, t_i)_{i \in [k]}$.
  
  \item[\textbf{Find:}] A set of $s_i$-$t_i$ walks of minimum total cost that jointly cover each vertex at least once.
\end{description}

The best known approximation algorithm for this problem in general metrics achieves an approximation ratio of $(1 + 2e^{-1/2}) < 2.214$ due to Böhm, Friggstad, Mömke and Spoerhase \cite{böhm2026approximatingtravelingsalesmanproblems}. For the graphic case this was recently improved to a 2-approximation due to Alimi, Dahlmeier, Mömke, Pabst and Vargas Koch \cite{graphic_multipath_paper}. In their paper, the authors also give a reduction showing that any better-than-two approximation for \Problem yields an immediate improvement for Graphic Multi-Path TSP as follows. First, compute a solution for \GraphicMultiPath using their LP sampling algorithm. Second, compute a solution for \Problem using $S = \bigcup_{i=1}^k s_i$ as the depots. Afterwards, connect each $s_i$ to $t_i$ using a shortest path and output the better of the two candidate solutions. Naively combining this with the result from \cref{thm:gmdm-tsp-general} yields a 1.992-approximation. However, this reduction is very lossy, and we improve on both the algorithm and the analysis. This leads to the following approximation guarantee.

\begin{restatable}{theorem}{thmmultipath}
\label{thm:multipath_improved}
		There is a polynomial-time algorithm that given an instance of Graphic Multi-Path TSP outputs a solution with cost of at most $\frac{11e-10}{6e-6}\cdot|\OPT\,| < 1.9304\cdot|\OPT\,|$.
\end{restatable}

We keep the sampling procedure unchanged and introduce a new second candidate solution: By contracting the components induced by the sampled paths, we obtain a \Problem instance whost optimal cost is at most that of the original instance. After solving the \Problem instance using \cref{thm:gmdm-tsp-general}, we only need to correct the parity of some vertices by computing an appropriate join. We then return the better of the two candidate solutions.

\subsection{Related Work}
 TSP and its variants have been studied extensively, from both a practical and theoretical standpoint. In this section, we only give results for variants closely related to \Problem and Graphic Multi-Path TSP and refer to \cite{Traub_Vygen_2024} for an in-depth overview of approximation results for other variants.

For \GraphicTSP, the first algorithm beating the approximation factor of $3/2$ of the Christofides algorithm was presented by Oveis Gharan, Saberi and Singh \cite{OveisGharanSaberiSingh2011GraphicTSP} in 2011. Whereas their result beats $3/2$ only by a tiny factor, the first sizable improvement of the approximation factor is due to Mömke and Svensson \cite{mömke_svensson}, who introduced the concept of removable pairings to obtain a $1.461$-approximation for Graphic TSP. This result was later improved to a $13/9$-approximation by Mucha \cite{mucha2012graphicTSP} and then to a $7/5$-approximation by Seb\H{o} and Vygen \cite{seboVygen2014graphic}.

A simplified version of the algorithm by Mömke and Svensson \cite{mömke_svensson} can be used to obtain a $4/3$-approximation for Graphic TSP on subcubic graphs. 
In \cref{sec:gmdmtsp_subcubic}, we show that we can match this guarantee for \Problem. 
A core idea in \Cref{sec:gmdmtsp_subcubic} is to complete an \emph{\evencover} to a full solution for \Problem. Similar approaches have been used previously to improve the approximation ratio for Graphic TSP on cubic (i.e., $3$-regular) graphs beyond $4/3$, first to a $9/7$-approximation by Dvo{\v{r}}{\'{a}}k, Král and Mohar \cite{dvorakKralMohar2017cubic} and later to a $5/4$-approximation due to Wigal, Yoo and Yu \cite{WigalYooYu2021cubic}.
\section{\texorpdfstring{A $3/2$-Approximation for \Problem}{A 3/2-Approximation for Graphic Multi-Depot TSP}}
\label{sec:gmdmtsp_general}

The main result of this section is a $3/2$-approximation for \Problem that runs in polynomial time for an arbitrary number of depots. Our algorithm is a slight adaptation of Christofides-Serdyukov: It first computes a forest spanning all non-depot vertices and then does a parity correction. In the analysis, we exploit savings in the cost of the forest to counterbalance additional costs in the parity correction step. This is based on counting edges and components, and thus crucially depends on the fact that we have a graphic metric.

\gmdtspgeneral*

\subsection{Description of the Algorithm} 
Let $\mathcal{C}$ denote the set of connected components of the input graph without depots, i.e., of $G[V \setminus D]$. We say an edge set is a \emph{depot-rooted spanning forest}
if it consists of one spanning tree for each component $C_i \in \C$ plus, for each $C_i$, one edge joining $C_i$ to a depot. A depot-rooted spanning forest can easily be computed by computing an arbitrary spanning tree of each component $C_i \in \mathcal{C}$ and attaching the spanning tree to an adjacent depot using one additional edge. This is always possible, as we assume the original graph to be connected.

The algorithm starts by computing a depot-rooted spanning forest $F$ followed by computing a minimum cardinality $\operatorname{odd}(F)$-join $J$ using the entire edge set of $G$. Finally, we return the multi-edge set $F \disjcup J$. An example run of the algorithm is depicted in \Cref{fig:christ}.

\begin{figure}[tbp]
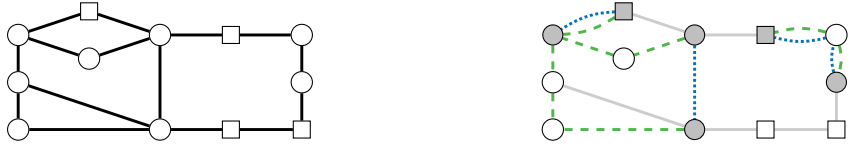

    \centering
    \begin{subfigure}[t]{0.495\textwidth}
        \ctikzfig{img/christofides_a}
        \caption{\Problem instance.}
    \end{subfigure}
    \hfill
    \begin{subfigure}[t]{0.495\textwidth}
        \ctikzfig{img/christofides_b}
        \caption{Depot-rooted spanning forest $F$ and $\odd(F)$-join.}
    \end{subfigure}
    \caption{$3/2$-approximation algorithm applied to an example \Problem instance shown in (a). Depots are drawn as squares. In (b), we compute a depot-rooted spanning forest $F$ shown as green dashed edges followed by a minimum cardinality $\odd(F)$-join in dotted blue edges where the vertices of $\odd(F)$ are drawn in gray.}\label{fig:christ}
\end{figure}

\subsection{Analysis}

In this section, we prove that the approximation guarantee of our algorithm is $\tfrac{3}{2}$. For the entirety of the analysis, we fix a \Problem instance $(G, D)$ where $G = (V,E)$, an optimal solution $\OPT$ and the solution computed by the algorithm $\ALG$ with the corresponding depot-rooted spanning forest $F$ and minimum cardinality $\operatorname{odd}(F)$-join $J$.

We introduce the following notation for the analysis. For a graph $H$, we write $\comp(H)$ for the number of connected components in $H$. Using this notation, let
\[
k = \comp(G[V \setminus D])= |\mathcal{C}| \qquad\text{and}\qquad k^* = \comp((V \setminus D, \OPT[V \setminus D])).
\]

Note that in a depot-rooted spanning forest each component $C_i \in \mathcal{C}$ is connected. For $\OPT$ this is not necessarily the case -- the graphs $(C_i, \OPT[C_i])$ might have multiple connected components.
The total number of times the optimal solution creates additional connected components over all $C_i \in \mathcal{C}$ is given by $(k^*-k)$, that is,
\begin{equation}
\label{eq:definition_ksub}
k^* - k = \sum_{C_i \in \mathcal{C}} \left(\comp((C_i, \OPT[C_i])) - 1
\right).
\end{equation}

This quantity will be important throughout the analysis. We analyze the cost of the depot-rooted spanning forest $F$ and $\odd(F)$-join $J$ separately starting with the forest. First, we give an upper bound on the size of $F$ compared to \OPT.

\begin{lemma}\label{lem:tree_cost}
    It holds that $|F| \leq |\OPT\,| - k^*$.
\end{lemma}

\begin{proof}
We have $|F| = |V \setminus D|$. Note that $(V \setminus D, \OPT[V \setminus D])$ has $k^*$ components and thus $\OPT$ contains at least $|V \setminus D| - k^*$ edges connecting non-depot vertices. Additionally, for each component of  $(V \setminus D, \OPT[V \setminus D])$ there must be at least two incident depot edges and thus $|\OPT| \geq (|V \setminus D| - k^*) + 2k^* = |F| + k^*.$
\end{proof}

We proceed by proving an upper bound on the cost of the parity correction step. Note that we cannot use the simple Christofides-Serdyukov bound of $\tfrac{1}{2}|\OPT|$ since odd-degree vertices of $F$ may be part of different components of $\OPT$. Instead, we use the notion of \emph{auxiliary edge sets} which was introduced by Xu and Rodrigues \cite{2010_1_5_approx_mdmtsp} in their work on the metric version of \ProblemNonGraphic.

\begin{definition}[auxiliary edge set]
    We call an edge set $A \subseteq E$ an auxiliary edge set (with respect to $F$ and $\OPT$) if each connected component of $(V, OPT \cup A)$ contains an even number of odd-degree vertices of $F$.
\end{definition}

When provided with an auxiliary edge set $A$, the following lemma from Xu and Rodrigues \cite{2010_1_5_approx_mdmtsp} allows us to upper bound the cost of the parity correction step. More precisely, when considering $\OPT \disjcup 2A$, the usual $\tfrac{1}{2}$-bound from the analysis of Christofides-Serdyukov holds.

\begin{lemma}[see \cite{2010_1_5_approx_mdmtsp}]
\label{lem:parity_correction_cost}
    Let $A$ be an auxiliary edge set with respect to $\OPT$ and $F$. We can upper bound the size of $J$ as
    $$
        |J| \leq \frac{1}{2} (|\OPT\,| + 2|A|).
    $$
\end{lemma}

Given an auxiliary edge set, we now have the necessary tools to bound the cost of the algorithm. It remains to show that a sufficiently small auxiliary edge set $A$ always exists. Note that since we only use the notion of auxiliary edge sets to upper bound the cost of an efficiently computable minimum cardinality $\operatorname{odd}(F)$-join, it suffices to show the existence of $A$, i.e., we do not need to show that $A$ can be computed efficiently. This is important since our construction of $A$ will make use of the component structure of $\OPT$. We use the following sufficient condition for auxiliary edge sets, also due to Xu and Rodrigues \cite{2017_xu_rodrigues}.

\begin{lemma}[see \cite{2017_xu_rodrigues}]
\label{lem:sufficient_condition_auxiliary_edge_set}
    Let $A \subseteq E$. If for every $\{u,v\} \in F$, the vertices $u$ and $v$ belong to the same connected component in $(V,\OPT \, \cup A)$, then $A$ is an auxiliary edge set with respect to $\OPT$ and $F$.
\end{lemma}

We now show the existence of a small auxiliary edge set.

\begin{lemma}\label{lem:constructing_auxiliary_edge_set}
    There exists an auxiliary edge set $A$ with respect to $F$ and $\OPT$ with $|A| \leq k^*$.
\end{lemma}
\begin{proof}
    To construct $A$, we add two different types of edges:
    \begin{enumerate}
        \item For each component $C_i \in \mathcal{C}$, we add the same edge to $A$ that the algorithm used to connect $C_i$ to a depot in $F$.
        \item For each component $C_i \in \mathcal{C}$, consider the set of components in $(C_i, \OPT[C_i])$. We add a set of edges $A_{C_i}$ to $A$ that ensures that $(C_i, \OPT[C_i] \cup A_{C_i})$ is connected. 
    \end{enumerate}
    There are $k$ edges of the first type (one for each component in $\mathcal{C}$). 
    For the second type, we have  $\operatorname{comp}((C_i, \OPT[C_i])) - 1$ many edges for each component $C_i$. Note that the type 2 edges exist since $C_i$ is connected in $G[V\setminus D]$. Summing up over all components and using \eqref{eq:definition_ksub}, we obtain that there are $k^*-k$ type 2 edges in $A$. This proves $|A| \leq k + (k^* - k) = k^*$. 
    
    It remains to show that $A$ is indeed an auxiliary edge set. To do so, we show the sufficient condition from \cref{lem:sufficient_condition_auxiliary_edge_set}.
    Let $\{u,v\} \in F$ be arbitrary. We distinguish the following two cases: Either $u$ and $v$ are both non-depot vertices or exactly one of the two vertices is a depot. This is sufficient since $F$ contains no edges between two depot vertices. In the first case, $u$ and $v$ are part of the same connected component $C_i \in \mathcal{C}$. By definition of the type 2 edges, $(C_i, \OPT[C_i] \cup A_{C_i})$ is connected and $u$ and $v$ are part of the same component in $(V, \OPT \cup A)$.
    If either $u$ or $v$ is a depot, then the edge $\{u,v\}$ is a type 1 edge and thus also in $A$. This concludes the proof.
\end{proof}

\begin{proof}[Proof of \cref{thm:gmdm-tsp-general}]
    First, note that the algorithm computes a feasible solution to the given \Problem instance: By construction, every non-depot vertex is connected to at least one depot in the depot-rooted spanning forest and the parity correction ensures that the output solution is indeed an Eulerian multi-edge set. The cost of the algorithm consists of the cost of the forest and the cost of the parity correction, that is, $|\ALG| = |F| + |J|$. We use Lemmas \ref{lem:tree_cost}, \ref{lem:parity_correction_cost} and \ref{lem:constructing_auxiliary_edge_set} to bound the costs of $F$ and $J$, respectively, and obtain
    \begin{align*}
        |\ALG| = |F| + |J| &\leq (|\OPT| - k^*) + \frac{1}{2}(|\OPT| + 2|A|) \\ &\leq (|\OPT| - k^*) + \frac{1}{2}(|\OPT| + 2k^*) = \frac{3}{2}|\OPT|. \qedhere
    \end{align*}
\end{proof}

Note that the guarantee of $\frac{3}{2}$ is tight. Observe that instances with $|D| = 1$ are \GraphicTSP instances and our algorithm is identical to Christofides-Serdyukov on such instances. Hence, since Christofides-Serdyukov is tight for \GraphicTSP, our algorithm is too.
\section{\texorpdfstring{$4/3$-Approximation on Subcubic Graphs}{4/3-Approximation on Subcubic Graphs}}
\label{sec:gmdmtsp_subcubic}

A natural idea for improving the approximation factor for \Problem from \cref{sec:gmdmtsp_general} is to apply techniques already known from \GraphicTSP. However, \Problem poses the additional challenge of first assigning the non-depot vertices to depots before then solving the \GraphicTSP instances induced by the assignment. To overcome this, we combine two different techniques: We use the notion of \emph{removable pairings} as introduced by Mömke and Svensson \cite{mömke_svensson} to compute a good \emph{\evencover} (see \Cref{def:even_cover}) that can be cheaply completed to a feasible solution for \Problem. 

The main result of this section is a $\tfrac{4}{3}$-approximation for \Problem when restricted to subcubic graphs, i.e., graphs of maximum degree 3. This matches the approximation guarantee for \GraphicTSP on the same graph class. 
However, we will see that we need to adapt the approach used for \GraphicTSP significantly. The main difference is that in contrast to \GraphicTSP, we cannot restrict ourselves to 2-vertex-connected instances. More precisely, given a cut-vertex $v$, we do not know in advance which component of $G-v$ should cover $v$.

\gmdmtspsubcubic*

For the rest of this section we fix a \Problem instance $(G, D)$ on a subcubic graph $G = (V, E)$ and denote by $\nn{E}$ the set of edges between two non-depot vertices, that is, $\nn{E} \coloneqq E(G[V \setminus D])$. Similarly, we write $\nd{E}$ for the set of edges connecting a non-depot and a depot vertex and write $\dd{E}$ for the set of edges between two depots. More generally, for an edge set $A \subseteq E$, we write $\nn{A} \coloneqq \nn{E} \cap A$, $\nd{A} \coloneqq \nd{E} \cap A$ and $\dd{A} \coloneqq \dd{E} \cap A$. Analogously, we write $\nn{V} \coloneqq V \setminus D$ and $\nn{G} \coloneqq (\nn{V}, \nn{E})$.

\subsection{Even Covers}

We denote a set of disjoint cycles as \evencover, following the definition in \cite{WigalYooYu2021cubic}. This will be an important object in our algorithm.

\begin{definition}[\evencover]\label{def:even_cover}
    We call an edge set $S \subseteq E$ (not allowing multi-edges) an \evencover if $\deg_S(v) \in \{ 0,2\}$ for all $v \in V$.
\end{definition}

Given an \evencover $S$ of $G$, we can complete it to a feasible solution by using double edges to ensure that every component in $(V, S)$ is connected to a depot.

\begin{lemma}\label{lem:completion_of_s}
    Given an \evencover $S \subseteq E$, we can efficiently compute a feasible solution $\completion{S}$ to \Problem of cost
    $$
        |\completion{S}| = |\nn{S}| + |\dd{S}| + 2 \comp((\nn{V}, \nn{S})).
    $$
\end{lemma}
\begin{proof}
    We start by counting the number of depot-free components of $(V, S)$. In $S$, every component of $(\nn{V},\nn{S})$ has zero or two incident edges to depots since $G$ is subcubic. Therefore, the number of depot-free components in $S$ is $\comp((\nn{V}, \nn{S})) - |\nd{S}| / 2$.
    Next, consider the graph obtained by contracting all components of $S$ in $G$. Moreover, contract all components which contain a depot to a single vertex. The resulting graph is connected since $G$ is connected and has a vertex for each depot-free component of $(\nn{V}, \nn{S})$ and one additional vertex for the depots. A spanning tree in the contracted graph corresponds to $ \comp((\nn{V}, \nn{S})) - |\nd{S}| / 2$ edges $T_S$ in the original graph $G$. We choose $\completion{S} = S \disjcup 2T_S$ to obtain the desired equality.
\end{proof}

We continue by deriving an upper bound on the cost of $\completion{S}$. To this end, we define a suitable cost function and then compute an even cover minimizing that cost function. The following lemma shows that the latter is efficiently possible for all cost functions using a simple $T$-join computation.

\begin{lemma}\label{lem:minimize_cost_function}
    For every cost function $c: E \to \R$, we can efficiently compute an \evencover $S$ which minimizes $c(S)$.
\end{lemma}
\begin{proof}
    Consider the set of negative cost edges $N = \{ e \in E \mid c(e) < 0\}$. Now, define $c'(e) = |c(e)|$ and compute a $\min$ cost $\odd(N)$-join $J$ in $G$ with respect to $c'$. Clearly, $S_c = N \triangle J$ is an even cover, and we claim that it minimizes cost w.r.t.\ the cost function $c$.
    To see this, observe that for any \evencover $S$, the set $N \triangle S$ is an $\odd(N)$-join and thus, $S$ can be written as $S = N \triangle (N \triangle S)$. Lastly, for any $\odd(N)$-join $J'$, we have $c(N \triangle J') = c(N) + c'(J')$
    and since we compute a $\min$-cost $\odd(N)$-join, we conclude that $c(S_c) \leq c(S)$.
\end{proof}

To define an appropriate cost function with which we bound $|\completion{S}|$, we recall the notion of removable pairings as introduced by Mömke and Svensson \cite{mömke_svensson}.

\begin{definition}[removable pairing]
    A removable pairing in a graph $(V, E)$ is a pair $(R, \P)$ with $R \subseteq E$ and $\P \subseteq { R \choose 2 }$ that fulfills the following properties:
        \begin{enumerate}[(i)]
            \item for each $P = \{e_1, e_2 \} \in \P$, there exists a vertex $v \in V$ with $\deg(v) \geq 3$ such that both $e_1$ and $e_2$ are incident to $v$; \label{enu:removable_pairings_2}
            \item the elements of $\P$ are pairwise disjoint; \label{enu:removable_pairings_3}
            \item for any set $F \subseteq R$ with $|F \cap P| \leq 1$ for all $P \in \P$, the graph $(V ,E \setminus F)$ is connected. \label{enu:removable_pairings_4}
        \end{enumerate}
        The elements of $R$ are called removable edges and the elements of $\P$ pairs. We say that an edge $e$ is paired if there exists some $P \in \P$ with $e \in P$.
\end{definition}

Hence, a removable pairing provides us with sets of edges $F$ which we can remove from $E$ without losing connectivity.
To construct a removable pairing, we use DFS trees as initially suggested by Mömke and Svensson for subcubic graphs. A visualization of the construction of the removable pairing is given in \cref{fig:subcubic_algo}a.
For each component $C \in \C$, we choose a root $r_C$ such that $r_C$ is adjacent to a depot. Next, we compute a DFS tree $T_C$ in $C$ starting at $r_C$ and mark each back-edge $e = \{v, u\}$ as removable. A back edge describes an edge that connects two vertices in the tree but is not itself a tree edge, i.e.,  all edges in $E(G[C]) \setminus T_C$. Moreover, we pair $e$ with the unique first tree edge $\{u,w\}$ on the $u$-$v$-path in $T_C$, if $u \neq r_C$. For each such pairing, we refer to $u$ as the center of the pairing. We denote the set of all centers by $\centers$. 
Since $G$ is subcubic, we pair each tree edge at most once and obtain a pair for each back-edge, with the exception of the possible back edge to $r_C$. Note that there is at most one such edge as we choose $r_C$ adjacent to a depot. We denote by $\ell \leq k$ the number of such ``bad'' roots, and conclude that for the resulting removable pairing $(R, \P$) it holds that
\begin{equation}\label{equ:num_removable_reformulated}
     |\P| = |\centers| = |\nn{E}| - |\nn{V}| + k - \ell \quad \text{ and } \quad |R| = 2(|\nn{E}| - |\nn{V}| + k) - \ell.
\end{equation}
In particular, we have $|R| - |\centers| = |\nn{E}| - |\nn{V}| + k$.

From now on, we fix the constructed removable pairing $(R, \P)$. Recall that $W$ denotes the set of centers of $\P$. As we will see in the subsequent \cref{lem:bound_on_completion_with_phi} and its proof, the following cost function $w_R: E \rightarrow \R$ will be helpful to upper bound the size of the completed solution $\completion{S}$ of an \evencover $S$. 
\begin{equation}
\label{def:w_e}
    w_R(e) =
\begin{cases}
  1 - |\centers \cap e|,  & e \in R, \\
  -1 - |\centers \cap e|, & e \in \nn{E} \setminus R, \\
  0,      & e \in \nd{E}, \\
  1           & e \in \dd{E}.
\end{cases}
\end{equation}

\begin{lemma}\label{lem:bound_on_completion_with_phi}
    Given an \evencover $S \subseteq E$, it holds that
    $
        |\completion{S}| \leq 2 |\nn{V}| + w_R(S).
    $
\end{lemma}
\begin{proof}
    We want to apply \Cref{lem:completion_of_s} to upper bound the cardinality of $F(S)$. To this end, we upper bound the number of components of $(\nn{V}, \nn{S})$. For this, we construct an edge set $L(S)$ such that $(\nn{V}, \nn{S} \cup L(S))$ consists of at most $k$ components by making use of our removable pairing. We start with all non-removable edges not yet in $\nn{S}$, that is, $(\nn{E} \setminus R) \setminus \nn{S}$, and add one paired tree edge for each center $w$ that is uncovered, i.e. $\delta(w)\cap S = \emptyset$.
    This yields $|L(S)| = (|\nn{E}| - |R| - |\nn{S} \setminus R|) + \uncovered(S)$ where $\uncovered(S)$ denotes the number of centers that are not covered by $S$. In particular, $(\nn{V}, \nn{S} \cup L(S))$ contains all non-removable edges and at least one removable edge per pairing.

    By property (iii) from the definition of removable pairings, every component of $(\nn{V}, \nn{E})$ is connected in $(\nn{V}, \nn{S} \cup L(S))$,
    and thus, this graph consists of at most $k$ components. It follows that $\comp((\nn{V}, \nn{S})) \leq k + |L(S)|$.
    To further understand $|L(S)|$, we first observe that $2\uncovered(S) = 2 |\centers| - \sum_{v \in \centers} \deg_S(v)$ since every covered center has degree 2 in $S$ and every uncovered center has degree $0$.
    Combining this with (\ref{equ:num_removable_reformulated}), we obtain 
    \begin{align*}
        2|L(S)| &= 2|\nn{E}| - 2|R| - 2|\nn{S} \setminus R| + 2\uncovered(S) \\
        &= 2|\nn{E}| - 2(|R| - |\centers|) - 2|\nn{S} \setminus R| - \sum_{v \in \centers} \deg_S(v) \\
        &\overset{(\ref{equ:num_removable_reformulated})}{=} 2|\nn{E}| - 2(|\nn{E}| - |\nn{V}| + k) - 2|\nn{S} \setminus R| - \sum_{v \in \centers} \deg_S(v) \\
        &= 2|\nn{V}| - 2k - 2|\nn{S} \setminus R| - \sum_{v \in \centers} \deg_S(v).
    \end{align*}
    Inserting this into \Cref{lem:completion_of_s}, we get the desired bound as follows.
    \begin{align*}
        |\completion{S}| &\leq |\nn{S}| + |\dd{S}| + 2 (k + |L(S)|) \\
        &= |\nn{S}| + |\dd{S}| + 2k + 2|\nn{V}| - 2k - 2|\nn{S} \setminus R| - \sum_{v \in \centers} \deg_S(v) \\
        &=  2|\nn{V}| + |\nn{S} \cap R|  - |\nn{S} \setminus R| + |\dd{S}| - \sum_{v \in \centers} \deg_S(v) \\
        &= 2|\nn{V}| + w_R(S).
    \end{align*}

    Note that the second to last line motivates the definition of $w_R$ and the last equality follows immediately from that definition.
\end{proof}

The results so far lead us to the following formulation of our algorithm: Start by computing a removable pairing using DFS trees in the components of $\nn{G}$ as described above. Using these DFS trees, define $w_R$, efficiently compute a minimum cost \evencover $S$ using \Cref{lem:minimize_cost_function} and complete this to a solution $F(S)$ using \Cref{lem:completion_of_s}. Pseudocode of our algorithm can be seen in \Cref{alg:4_3_gmdm_tsp}, and an example run of the algorithm is shown in \Cref{fig:subcubic_algo}. It remains to show the existence of an \evencover $S$ such that $w_R(S)$ is sufficiently small.

\begin{figure}[thb]
    \centering
    \subcaptionbox{Computation of the removable pairing $(R, \P)$}[0.48\textwidth]{
        \begin{tikzpicture}[reset cm, scale = 1.05]
    \node[style=Depot] (A) at (0,0) {};
    \node[style=Normal Node] (B) at (0,1) {};
    \node[style=Normal Node] (C) at (1,1) {};
    \node[style=Normal Node] (D) at (2,1) {};
    \node[style=Depot] (E) at (3,1) {};
    \node[style=Normal Node] (F) at (4,1) {};
    \node[style=Center] (G) at (4.7,1) {};
    \node[style=Normal Node] (H) at (3,2) {};
    \node[style=Depot] (I) at (4,2) {};
    \node[style=Normal Node] (J) at (5.2,2) {};
    \node[style=Center] (K) at (1,0) {};
    \node[style=Normal Node] (L) at (1.5,-.7) {};
    \node[style=Normal Node] (M) at (2,0) {};
    \node[style=Depot] (N) at (3,0) {};
    \node[style=Depot] (O) at (4,0) {};
    \node[style=Normal Node] (P) at (5.2,0) {};

    \draw[style=Tree] (B) -- (C) -- (D) -- (H);
    \draw[style=Tree] (C) -- (K) -- (M) -- (L);
    \draw[style=Tree] (F) -- (G) -- (J);
    \path (J) edge[style=Tree, bend left = 15] (P);

    \draw[style=Normal Edge] (A) -- (B);
    \path (B) edge[style=Normal Edge, bend left = 10] (H);
    \draw[style=Normal Edge] (H) -- (I) -- (J);
    \draw[style=Normal Edge] (K) -- (L);
    \draw[style=Normal Edge] (D) -- (E) -- (F);
    \draw[style=Normal Edge] (M) -- (N) -- (O) -- (P) -- (G);

    \node[style=Label] (l1) at (0, 1.25) {$r_1$}; 
    \node[style=Label] (l2) at (4, 1.25) {$r_2$}; 

    \begin{pgfonlayer}{background}
      \parallelEdge[style=Removable1, bend left = 10]{B}{H}{2pt}
      \parallelEdge[style=Removable1]{K}{M}{-2pt}
      \parallelEdge[style=Removable1]{K}{L}{2pt}
      \parallelEdge[style=Removable1]{G}{J}{-2pt}
      \parallelEdge[style=Removable1]{G}{P}{2pt}
    \end{pgfonlayer}
    \path (1.5, -.45) edge[style=PairingArrow] (1.75, -.15);              
    \path (5, 1.3) edge[style=PairingArrow, bend left = 15] (5, .7);      
\end{tikzpicture}
    \label{subfig:removable_pairings}
    }
    \subcaptionbox{Cost function $w_R$ and min-cost $\text{odd}(N)$-join $J$}[.48\textwidth]{
        \begin{tikzpicture}[reset cm, scale = 1.05]
    \node[style=Depot] (A) at (0,0) {};
    \node[style=Gray] (B) at (0,1) {};
    \node[style=Gray] (C) at (1,1) {};
    \node[style=Normal Node] (D) at (2,1) {};
    \node[style=Depot] (E) at (3,1) {};
    \node[style=Gray] (F) at (4,1) {};
    \node[style=Gray] (G) at (4.7,1) {};
    \node[style=Gray] (H) at (3,2) {};
    \node[style=Depot] (I) at (4,2) {};
    \node[style=Gray] (J) at (5.2,2) {};
    \node[style=Gray] (K) at (1,0) {};
    \node[style=Gray] (L) at (1.5,-.7) {};
    \node[style=Gray] (M) at (2,0) {};
    \node[style=Depot] (N) at (3,0) {};
    \node[style=Depot] (O) at (4,0) {};
    \node[style=Gray] (P) at (5.2,0) {};

    \draw[style=Normal Edge] (A) --node[style=Label, right, inner sep = 0.05cm] {$0$} (B);
    \draw[style=Join] (B) --node[style=Label, above, inner sep = 0.05cm] {$-1$} (C);
    \draw[style=Normal Edge] (C) --node[style=Label, above, inner sep = 0.05cm] {$-1$} (D);
    \draw[style=Join] (D) --node[style=Label, above, inner sep = 0.05cm] {$0$} (E);
    \draw[style=Join] (E) --node[style=Label, above, inner sep = 0.05cm] {$0$} (F);
    \draw[style=Normal Edge] (F) --node[style=Label, above, inner sep = 0.05cm] {$-2$} (G);
    \path (B) edge[style=Normal Edge, bend left = 10] (H);
    \node[style=Label] (lbh) at (1.5, 1.85) {$1$};
    \draw[style=Join] (D) --node[style=Label, left, inner sep = 0.08cm] {$-1$} (H);
    \draw[style=Normal Edge] (H) --node[style=Label, above, inner sep = 0.05cm] {$0$} (I);
    \draw[style=Normal Edge] (I) --node[style=Label, above, inner sep = 0.05cm] {$0$} (J);
    \draw[style=Join] (G) --node[style=Label, left, inner sep = 0.05cm] {$0$} (J);
    \draw[style=Normal Edge] (G) --node[style=Label, left, inner sep = 0.05cm] {$0$} (P);
    \path (J) edge[style=Normal Edge, bend left = 15] (P);
    \node[style=Label] (ljp) at (5.17, 1) {$-1$};
    \draw[style=Join] (O) --node[style=Label, above, inner sep = 0.05cm] {$0$} (P);
    \draw[style=Join] (N) --node[style=Label, above, inner sep = 0.05cm] {$1$} (O);
    \draw[style=Join] (M) --node[style=Label, above, inner sep = 0.05cm] {$0$} (N);
    \draw[style=Normal Edge] (C) --node[style=Label, left, inner sep = 0.05cm] {$-2$} (K);
    \draw[style=Normal Edge] (K) --node[style=Label, above, inner sep = 0.05cm] {$0$} (M);
    \draw[style=Join] (K) --node[style=Label, left, inner sep = 0.05cm] {$0$} (L);
    \draw[style=Normal Edge] (L) --node[style=Label, right, inner sep = 0.05cm] {$-1$} (M);
\end{tikzpicture}
    }
    \vspace*{0.15cm}

    \subcaptionbox{Even cover $S^* = N \Delta J$ and completion $F(S^*)$}[0.48\textwidth]{
        \begin{tikzpicture}[reset cm, scale = 1.05]
    \node[style=Depot] (A) at (0,0) {};
    \node[style=Normal Node] (B) at (0,1) {};
    \node[style=Normal Node] (C) at (1,1) {};
    \node[style=Normal Node] (D) at (2,1) {};
    \node[style=Depot] (E) at (3,1) {};
    \node[style=Normal Node] (F) at (4,1) {};
    \node[style=Normal Node] (G) at (4.7,1) {};
    \node[style=Normal Node] (H) at (3,2) {};
    \node[style=Depot] (I) at (4,2) {};
    \node[style=Normal Node] (J) at (5.2,2) {};
    \node[style=Normal Node] (K) at (1,0) {};
    \node[style=Normal Node] (L) at (1.5,-.7) {};
    \node[style=Normal Node] (M) at (2,0) {};
    \node[style=Depot] (N) at (3,0) {};
    \node[style=Depot] (O) at (4,0) {};
    \node[style=Normal Node] (P) at (5.2,0) {};

    \draw[style=EvenCover] (C) -- (D) -- (E) -- (F) -- (G) -- (J);
    \path (J) edge[style=EvenCover, bend left = 15] (P);
    \draw[style=EvenCover] (P) -- (O) -- (N) -- (M) -- (L) -- (K) -- (C);

    \path (A) edge[style=DoubleEdges, bend left = 15] (B);
    \path (B) edge[style=DoubleEdges, bend left = 15] (A);
    \path (D) edge[style=DoubleEdges, bend left = 15] (H);
    \path (H) edge[style=DoubleEdges, bend left = 15] (D);

    \draw[style=UnusedEdge] (B) -- (C);
    \draw[style=UnusedEdge] (K) -- (M);
    \draw[style=UnusedEdge] (H) -- (I) -- (J);
    \draw[style=UnusedEdge] (G) -- (P);
    \path (B) edge[style=UnusedEdge, bend left = 10] (H);

\end{tikzpicture}
    }
    \caption{Depots are shown as squares and non-depots are shown as circles. In (a), we compute a removable pairing $(R, \P)$ using the DFS trees marked by green, dashed edges. Removable edges are shown dotted in red, paired edges are indicated by arrows, and centers are shown in red. In particular, the back edge to $r_1$ is not paired. In (b), we give the edge weights $w_R$. The vertex set $\text{odd}(N)$ is shown in gray, and a minimum cost $\text{odd}(N)$-join $J$ for cost $|w_R|$ is indicated by blue, dotted edges. Lastly, in (c), we give the even cover $S^* = N \Delta J$ (purple, dashed) and its completion to a feasible solution using the orange, dotted double edges.
    }
    \label{fig:subcubic_algo}
\end{figure}
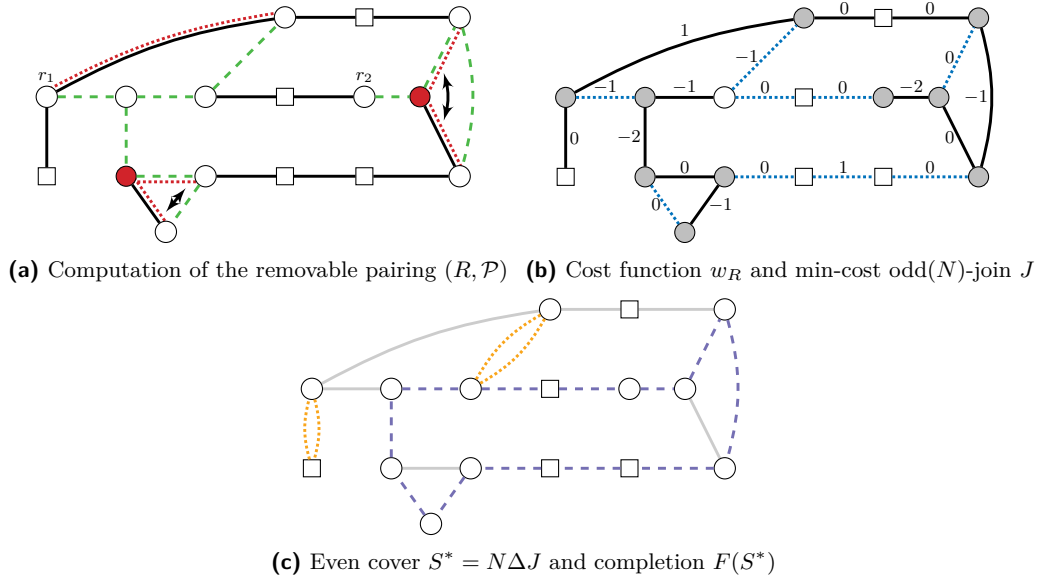

\begin{algorithm}[htb]
    \caption{Removable Pairing based Approximation Algorithm for \Problem}\label{alg:4_3_gmdm_tsp}
    \textbf{Input} Graph $G = (V,E)$, depots $D \subseteq V$
    \begin{algorithmic}[1]

    \vspace*{0.15cm}
    \Statex \textbf{Compute removable pairing $(R, \P)$.}
    \For{$C \in \C$}
        \State Find a vertex $r_C \in C$ adjacent to a depot.
        \State Compute a DFS tree $T_C$ rooted in $r_C$.
        \State{Compute a removable pairing $(R_C, \P_C)$ in $G[C]$ by marking every back-edge $(v,u)$ as removable and pairing it with the first edge $(u,w)$ on the $u$-$v$-path in $T_C$ if $u \neq r_C$.}
    \EndFor
    \State Combine the $R_C$ and $P_C$ to a removable pairing $(R, \P)$.

    \vspace*{0.15cm}
    \Statex \textbf{Compute \evencover $S_0$.}
    \State Define $w_R$ as in \Cref{def:w_e}.
    \State Let $S_0$ be the \evencover minimizing $w_R(S_0)$ obtained with \Cref{lem:minimize_cost_function}.

    \vspace*{0.15cm}
    \Statex \textbf{Complete $S_0$ to a feasible solution.}
    \State Let $\completion{S_0}$ be the solution obtained by applying \Cref{lem:completion_of_s} to $S_0$.
    \vspace*{0.15cm}
    
    \State \Return $\completion{S_0}$
    \end{algorithmic}
\end{algorithm}

\subsection{Existence of Sufficiently Cheap Even Covers}

To show the existence of an \evencover $S$ such that $w_R(S)$ is sufficiently small, we construct a comparison graph $K$ with respect to a fixed optimal solution $\OPT$. We denote by $\double$ the number of double edges that $\OPT$ uses in $\nn{E}$. Using the comparison graph $K$, we obtain an upper bound for the cost of our algorithm in \Cref{lem:sufficiently_cheap_even_cover}. In the subsequent \Cref{lem:opt_lower_bound} we show that the upper bound proved in \Cref{lem:sufficiently_cheap_even_cover} is at most $\frac{4}{3}|\OPT|$.
Our proof of \Cref{lem:sufficiently_cheap_even_cover} uses the following well known result. For the sake of completeness, we prove this result in Appendix \ref{sec:appendix_sec_3}. 

\begin{restatable}[see e.g.~proof of Theorem 12.3 in \cite{Traub_Vygen_2024}]{proposition}{existenceofcheapjoins}
\label{pro:existence_of_cheap_joins}
    Let $H$ be a $2$-edge-connected graph. For arbitrary edge weights $c: E(H) \to \R$, there is an $\odd(H)$-join $J \subseteq E(H)$ such that
    $
        c(J) \leq \tfrac{1}{3} c(E(H)).
    $
\end{restatable}

By \Cref{lem:bound_on_completion_with_phi}, the trivial solution obtained from $S = \emptyset$ costs $2|\nn{V}|$, and a non-empty \evencover $S$ changes this cost by $w_R(S)$. Note that $w_R(\nn{E})$ is roughly $-|\nn{V}|$.
In a $2$-edge-connected graph, \Cref{pro:existence_of_cheap_joins} yields an \evencover that keeps at least two thirds of these savings. Thus, our total cost is about $2|\nn{V}| - \frac{2}{3}|\nn{V}| = \frac{4}{3}|\nn{V}|$.
This approach works well on the 2-connected blocks of $K$, but fails for the bridges of $K$. However, we can exploit that bridges of $K$ also incur additional cost for $\OPT$. Recall that $ k^* = \comp((V \setminus D, \OPT[V \setminus D])).$

\begin{lemma}\label{lem:sufficiently_cheap_even_cover}
    There exists an \evencover $S$ such that
    $$
        2|\nn{V}| + w_R(S) \leq \frac{4}{3} (|\nn{V}| + \double + k^* + |\dd{\OPT}|).
    $$
\end{lemma}
\begin{proof}
    For this proof, a \emph{block} of a graph refers to either a maximal $2$-(vertex-)connected subgraph or a bridge of that graph.
    We begin by constructing the comparison graph $K$, then show the existence of sufficiently cheap \evencovers w.r.t. $w_R$ on the blocks of $K$ with the goal of combining them to a global \evencover of the original graph. To find cheap \evencovers, we use $T$-joins and \Cref{pro:existence_of_cheap_joins}. 
    Denoting by $\OPT_1$ the edges our optimal solution $\OPT$ uses with multiplicity one, we define the comparison graph
    \[K \coloneqq (V, \nn{E} \cup (\OPT_1 \cap \nd{E}) \cup \dd{\OPT}).\]
    We want to find cheap \evencovers on the 2-connected blocks of $K$. Let $\bridges{K}$ denote the set of bridges of $K$. By definition these are not 2-connected. Observe that in subcubic graphs the $2$-connected blocks are vertex-disjoint and connected via the bridges of the graph.
    
    Let $B$ be a $2$-connected block in $K$ and let $J_B$ be a minimum $\odd(E(B))$-join in $B$ with respect to $-w_R$. Now, $S_B = E(B) \setminus J_B$ is an \evencover by construction. By \Cref{pro:existence_of_cheap_joins}, we conclude $w_R(S_B) \leq \frac{2}{3} w_R(E(B))$. Denote by $S$ the union of the \evencovers over all $2$-connected blocks which is clearly again an \evencover, and it fulfills
    \begin{equation}
    \label{eq:final_cost_S}
        2|\nn{V}| + w_R(S) \leq 2|\nn{V}| + \frac{2}{3} w_R(E(K) \setminus \bridges{K}).
    \end{equation}
    
    To bound $w_R(E(K) \setminus \bridges{K})$, observe that $\nn{E} \subseteq E(K)$. Hence, for every pairing $P \in \P$, the paired edges are always part of a common 2-connected block $B$ in $K$. However, the third edge incident to the center of $P$ might not be part of $B$. If this happens we say that $P$ is a \emph{bad pairing}, and the third edge must be a bridge of $K$. Hence, the number of bad pairings is at most $|\bridges{K}|$. For bad pairings $P\in\P$, their center is incident to 2 edges in $E(K) \setminus \bridges{K}$ and for all other pairings, their center is incident to 3 edges in $E(K) \setminus \bridges{K}$. 
    With this and \cref{equ:num_removable_reformulated} we obtain 
    \begin{align*}
        w_R(E(K) \setminus \bridges{K}) &\leq |R| - (|\nn{E}| - |R| - |\bridges{K}|) - (3|\centers| - |\bridges{K}|) + |\dd{\OPT}| \\
        &= 2(|R| - |\centers|) - |\centers| - |\nn{E}| + 2 |\bridges{K}| + |\dd{\OPT}| \\
        &\overset{(\ref{equ:num_removable_reformulated})}{=} 2(|\nn{E}| - |\nn{V}| + k) - (|\nn{E}| - |\nn{V}| + k - \ell) \\ 
        &\qquad - |\nn{E}| + 2 |\bridges{K}| + |\dd{\OPT}| \\
        &= - |\nn{V}| + k + \ell + 2 |\bridges{K}| + |\dd{\OPT}|.
    \end{align*}

    We can tighten this bound further by bounding the number of bridges $|\bridges{K}|$. Bridges in $K$ are either also in $\OPT$, in which case they have multiplicity two in $\OPT$ and are thus bounded by $\double$ or they are unused in $\OPT$ in which case they are also bridges in $\nn{G}$. The former holds since $\OPT_1 \subseteq E(K)$ by construction. By removing each such bridge from $\nn{G}$, we obtain at most $k^*$ components, hence at most $k^* - k$ such bridges. That is, $|\bridges{K}| \leq \double + (k^* - k)$. Plugging this into the previous inequality, we obtain
        $ w_R(E(K) \setminus \bridges{K}) \leq 2k^* + 2\double - |\nn{V}| - k + \ell + |\dd{\OPT}|.$
    Inserting this in \eqref{eq:final_cost_S} and recalling that $\ell \leq k$, we get the desired bound
    \begin{align*}
    2|\nn{V}| + \frac{2}{3} w_R(E(K) \setminus \bridges{K}) &\leq 2|\nn{V}| + \frac{2}{3}(2k^* + 2\double - |\nn{V}| - k + \ell + |\dd{\OPT}|) \\
    &\leq \frac{4}{3}(|\nn{V}| + k^* + \double) - \frac{2}{3}k + \frac{2}{3}\ell + \frac{2}{3}|\dd{\OPT}| \\
    &\leq \frac{4}{3}(|\nn{V}| + k^* + \double) + \frac{2}{3}|\dd{\OPT}| \\
    &\leq \frac{4}{3}(|\nn{V}| + k^* + \double + |\dd{\OPT}|). \qedhere
    \end{align*}
\end{proof}

Now, we connect the cost obtained in the previous lemma to $\OPT$.

\begin{lemma}\label{lem:opt_lower_bound}
    It holds that $|\OPT\,| \geq |\nn{V}| + \double + k^* + |\dd{\OPT}|.$
\end{lemma}
\begin{proof}
    First, observe that the support of $\nn{\OPT}$ has at least $|\OPT_1 \cap \nn{\OPT}| + \double \geq |\nn{V}| - k^*$ edges,
    where the last inequality holds since $k^* = \comp((\nn{V}, \nn{\OPT}))$ by definition. Lastly, since every component in $(\nn{V}, \nn{\OPT})$ must have two adjacent edges in $\nd{\OPT}$, we obtain that
    \begin{equation*}
        |\OPT| \geq |\OPT_1 \cap \nn{\OPT}| + 2\double + 2k^* + |\dd{\OPT}|\geq |\nn{V}| + \double + k^* + |\dd{\OPT}|. \qedhere
    \end{equation*}
\end{proof}

Our final algorithm is \Cref{alg:4_3_gmdm_tsp}; its analysis is straightforward now.
\begin{proof}[Proof of \Cref{thm:gmdm-tsp-subcubic}]
    \Cref{alg:4_3_gmdm_tsp} has polynomial running time as we can efficiently compute a $w_R$-minimizing \evencover by \cref{lem:minimize_cost_function}. 
    Combining Lemmas \ref{lem:sufficiently_cheap_even_cover} and \ref{lem:opt_lower_bound}, we obtain
\[
    |\ALG| = |\completion{S_0}| \leq 2|\nn{V}| + w_R(S_0) \leq \frac{4}{3} (|\nn{V}| + \double + k^* + |\dd{\OPT}|) \leq \frac{4}{3} |\OPT|. \qedhere
\]
\end{proof}
\section{A 1.9304-Approximation for Graphic Multi-Path TSP}
\label{sec:graphic_multpath}
In this section, we apply \cref{thm:gmdm-tsp-general} to improve the best known approximation factor for Graphic Multi-Path TSP. The goal of the section is to prove the following theorem.

\thmmultipath*

For ease of presentation, we describe a randomized version of our algorithm in this section. However, this algorithm can be derandomized similarly to \cite{graphic_multipath_paper} using a standard conditional expectation argument. A more detailed description of the derandomization can be found in Appendix \ref{app:derandomization}.

Our algorithm works as follows: At first, we follow the approach of \cite{graphic_multipath_paper} and sample paths from an optimal solution of an LP-relaxation of Graphic Multi-Path TSP. However, the union of these sampled paths may not cover all vertices, and may thus be infeasible. Because of this, we need a \emph{reconnection step} that ensures coverage of all vertices. We use two different reconnection approaches and output the cheaper of the two. The first is identical to \cite{graphic_multipath_paper} and simply reconnects each unconnected vertex greedily for a cost of $2$ per vertex. The second uses the sampled paths to build an auxiliary instance of \Problem and uses \cref{thm:gmdm-tsp-general} to compute a $3/2$-approximate solution for this instance. This solution is then augmented to a feasible solution for Graphic Multi-Path TSP in a parity correction step.

\subsection{Sampling Paths}
\label{sec:sampling_step}

The sampling procedure of the paths follows Alimi et al.~\cite{graphic_multipath_paper} and we refer to the full paper for more details. The algorithm first bidirects every edge in the graph and solves a multi-commodity flow LP-relaxation (see Appendix \ref{LP}). The algorithm then decomposes the flow of each commodity into paths and cycles and independently samples one path per commodity with probability equal to the weight coefficient of that path in the decomposition. We denote by $z_v$ the coverage of $v$ by flow, i.e., the total multi-commodity flow \emph{leaving} $v$. Similarly, $z_v^P$ denotes the \emph{coverage of $v$ by paths}, more precisely the flow leaving $v$ after discarding all cycles of the path-and-cycle decomposition.

We can split the expected path sampling cost among the vertices by making every vertex $v$ responsible for paying all edges that leave $v$. With this, $\mathbb{E}[\text{sampling cost of $v$}] = z_v^P$. \\

\noindent\textbf{Notation used in the analysis.} We now introduce several quantities that will be important for our improved analysis. First, we write $S$ and $T$ for the sets of all sources and sinks, respectively. For the remainder of this section, for some collection of sampled paths $P_1,\ldots,P_k$ we denote by
$\P = P_1 \, \dot{\cup} \, \ldots \, \dot{\cup} \, P_k$
the multi-set containing all the sampled paths. For $v\in V$ we define $\ell_v$ to be the number of edges of $\P$ leaving $v$.
Hence, we have
\begin{equation}
\label{eq:bounds_lv}
|\P| = \sum_{v\in V} \ell_v \quad\text{and}\quad \mathbb{E}[\ell_v] = z_v^P \quad \forall v \in V.
\end{equation}

Furthermore, we denote by $q_v$ the probability that there is no sampled path containing $v$, i.e., $v$ is not covered
\begin{equation}
\label{eq:bounds_qv}
    q_v \coloneqq \mathbb{P}[\text{$v$ not covered by any path}].
\end{equation}
We now introduce the quantities $\varphi_v$ for $v \in V$ and $\Phi$ which we will use to compare the two candidate solutions. For $v \in V$ let
\[\varphi_v \coloneqq \begin{cases}z_v^P + (1-q_v), & v \not\in S \cup T, \\ 3z_v^P, & v \in S\cup T,\end{cases} \quad\text{and}\quad \Phi \coloneqq \sum_{v\in V} \varphi_v.\]
Intuitively, if a large part of the LP-mass in the path-and-cycle decomposition is on paths, both $z_v^P$ and the probability $(1-q_v)$ that $v$ is covered by $P$ must be large for many vertices, and thus $\Phi$ is large. In this case, the first candidate solution gives a guarantee of $2|\OPT| - \tfrac{e-2}{2e-1}\Phi$ which is better than $2|\OPT|$ for large values of $\Phi$. Otherwise, the second candidate solution gives a guarantee of $\tfrac{3}{2}|\OPT| + \Phi$ which is better than $2|\OPT|$ for small values of $\Phi$.

\subsection{First Candidate for Reconnection}
\label{sec:first_candidate}
To obtain the first candidate solution $F_1$, we follow the approach of \cite{graphic_multipath_paper} and simply reconnect every vertex that is not covered by the union of the sampled paths for a cost of two. Clearly, since the underlying graph is unweighted and connected, this is possible using a simple greedy approach. We denote the resulting candidate solution by $F_1$.

\begin{restatable}{lemma}{lemmacostfone}
\label{lem:cost_F1}
		For the solution $F_1$ it holds that $\mathbb{E}[|F_1|] \leq 2|\OPT\,| - \tfrac{e-2}{2e-1}\Phi$.
\end{restatable}
The proof bounds the cost per vertex, making every vertex pay for its leaving edges and using $ \sum_v z_v \leq |OPT|$. 
While \cite{graphic_multipath_paper} used a bound of $2z_v$ on the cost per vertex, here we save $\gamma \varphi_v$ compared to that, where $\gamma = (e-2)/(2e-1)$. 
For vertices $v \in S \cup T$ this is straightforward, while for vertices $v \notin S \cup T$ we do a case distinction between $z_v^P \leq 1$ and $z_v^P \geq 1$. 
We defer the full proof to Appendix \ref{proof:cost_f1}.

\subsection{Second Candidate for Reconnection}
\label{sec:second_candidate}

We construct the second candidate solution $F_2$ as follows: Beginning with the union of all sampled paths $\P$, we contract all components of $\P$ into single vertices, and denote the set of these vertices by $D'$ and the resulting graph by $G'$. We apply the $3/2$-approximation for \Problem on the contracted instance $I' = (G', D')$ to compute a solution $ALG(I')$. We claim that $|\OPT(I')|\leq |\OPT(I)|$ where $|\OPT(I')|$ is the optimal objective value of $I'$ seen as an instance of \Problem.

\newpage
\begin{observation}
\label{obs:obj_value_aux_instance}
It holds that $|\OPT(I')| \leq |\OPT(I)|$.
\end{observation}
Indeed, consider any optimal solution $F$ for the original instance $I$ and let $F^*$ be the same solution after contracting $\P$. The contraction identifies the source-sink pairs so each walk in $F$ becomes a closed walk through a depot in $I'$. These walks cover all vertices and contraction cannot increase their total cost and hence $|F^*| \leq |F|$.

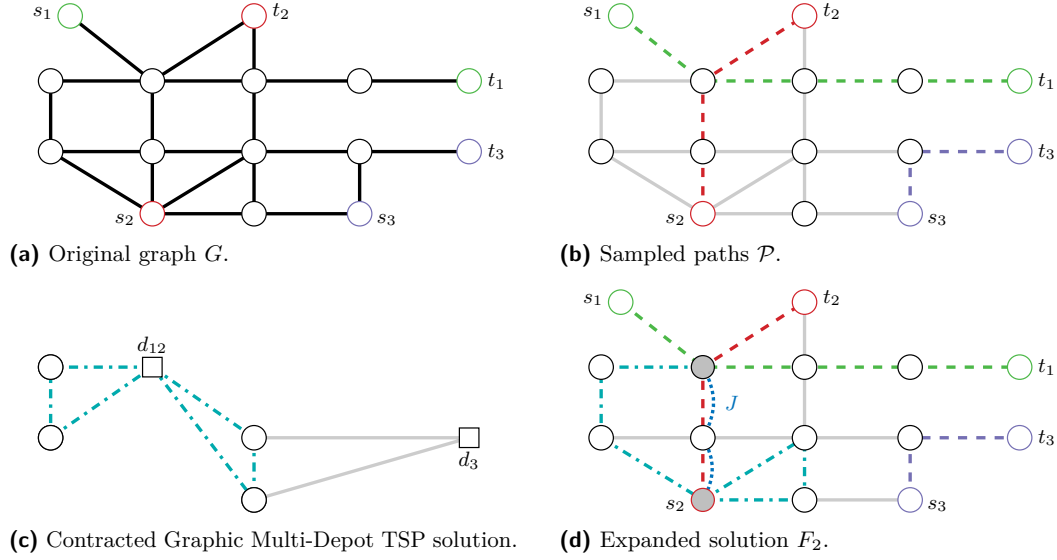
\begin{figure}[!t]
    \centering
    \captionsetup[subfigure]{hypcap=true}
    \begin{minipage}{1.0\linewidth}
        \centering
        \begin{subfigure}[t]{0.48\linewidth}
            \centering
            \resizebox{\linewidth}{!}{
                \begin{tikzpicture}[reset cm, yscale=0.65, xscale=0.9]
\def\fTwoCoordinates{%
  \coordinate (a) at (.50,2.20);
  \coordinate (b) at (1.80,2.20);
  \coordinate (c) at (3.10,2.20);
  \coordinate (d) at (4.45,2.20);
  \coordinate (e) at (.50,.95);
  \coordinate (f) at (1.80,.95);
  \coordinate (g) at (3.10,.95);
  \coordinate (h) at (4.45,.95);
  \coordinate (i) at (3.10,-.15);
  \coordinate (s1) at (.75,3.35);
  \coordinate (t1) at (5.85,2.20);
  \coordinate (s2) at (1.80,-.15);
  \coordinate (t2) at (3.10,3.35);
  \coordinate (s3) at (4.45,-.15);
  \coordinate (t3) at (5.85,.95);
}
\def\fTwoBaseGraph{%
  \foreach \from/\to in {
    a/b,a/e,e/f,f/g,b/c,c/d,c/g,c/t2,e/s2,g/s2,g/i,i/s2,
    s1/b,d/t1,f/s2,b/f,b/t2,h/s3,h/t3,i/s3,g/h}
    \draw[Normal Edge] (\from)--(\to);
}
\def\fTwoSampledPaths{%
  \draw[Normal Edge,draw=rwth-green, densely dotted] (s1)--(b)--(c)--(d)--(t1);
  \draw[Removable1,dashed] (s2)--(f)--(b)--(t2);
  \draw[Normal Edge,draw=rwth-purple,dashdotted] (s3)--(h)--(t3);
}
\def\fTwoLiftedTours{%
  \draw[Normal Edge, dashdotted, draw=rwth-turquoise] (b)--(a)--(e)--(s2);
  \draw[Normal Edge, dashdotted, draw=rwth-turquoise] (s2)--(g)--(i)--cycle;
}
\def\fTwoOriginalNodes{%
  \foreach \vertex in {a,b,c,d,e,f,g,h,i}
    \node[Normal Node] at (\vertex) {};
  \foreach \vertex in {s1,t1}
    \node[Normal Node,draw=rwth-green] at (\vertex) {};
  \foreach \vertex in {s2,t2}
    \node[Normal Node,draw=rwth-red] at (\vertex) {};
  \foreach \vertex in {s3,t3}
    \node[Normal Node,draw=rwth-purple] at (\vertex) {};
  \node[Label,anchor=south,inner sep=1pt] at ($(s1)+(-0.35,-0.1)$) {$s_1$};
  \node[Label,anchor=west,inner sep=1pt] at (6.05,2.20) {$t_1$};
  \node[Label,anchor=north,inner sep=1pt] at ($(s2)+(-0.35,0.05)$) {$s_2$};
  \node[Label,anchor=south,inner sep=1pt] at ($(t2)+(0.35,-0.1)$) {$t_2$};
  \node[Label,anchor=north,inner sep=1pt] at ($(s3)+(0.35,0.05)$) {$s_3$};
  \node[Label,anchor=west,inner sep=1pt] at (6.05,.95) {$t_3$};
}
  \fTwoCoordinates
  \fTwoBaseGraph
  \fTwoOriginalNodes
\end{tikzpicture}%
            }
            \caption{Original graph $G$.}
            \label{fig:1.93-a}
        \end{subfigure}
        \hfill
        \begin{subfigure}[t]{0.48\linewidth}
            \centering
            \resizebox{\linewidth}{!}{
                \begin{tikzpicture}[reset cm, yscale=0.65, xscale=0.9]
\def\fTwoCoordinates{%
  \coordinate (a) at (.50,2.20);
  \coordinate (b) at (1.80,2.20);
  \coordinate (c) at (3.10,2.20);
  \coordinate (d) at (4.45,2.20);
  \coordinate (e) at (.50,.95);
  \coordinate (f) at (1.80,.95);
  \coordinate (g) at (3.10,.95);
  \coordinate (h) at (4.45,.95);
  \coordinate (i) at (3.10,-.15);
  \coordinate (s1) at (.75,3.35);
  \coordinate (t1) at (5.85,2.20);
  \coordinate (s2) at (1.80,-.15);
  \coordinate (t2) at (3.10,3.35);
  \coordinate (s3) at (4.45,-.15);
  \coordinate (t3) at (5.85,.95);
}
\def\fTwoBaseGraph{%
  \foreach \from/\to in {
    a/b,a/e,e/f,f/g,c/g,c/t2,e/s2,g/s2,g/i,i/s2,i/s3,g/h}
    \draw[UnusedEdge] (\from)--(\to);
}
\def\fTwoSampledPaths{%
  \draw[Normal Edge,draw=rwth-green, dashed] (s1)--(b)--(c)--(d)--(t1);
  \draw[Removable1,dashed] (s2)--(f)--(b)--(t2);
  \draw[Normal Edge,draw=rwth-purple,dashed] (s3)--(h)--(t3);
}
\def\fTwoLiftedTours{%
  \draw[Normal Edge, dashdotted, draw=rwth-turquoise] (b)--(a)--(e)--(s2);
  \draw[Normal Edge, dashdotted, draw=rwth-turquoise] (s2)--(g)--(i)--cycle;
}
\def\fTwoOriginalNodes{%
  \foreach \vertex in {a,b,c,d,e,f,g,h,i}
    \node[Normal Node] at (\vertex) {};
  \foreach \vertex in {s1,t1}
    \node[Normal Node,draw=rwth-green] at (\vertex) {};
  \foreach \vertex in {s2,t2}
    \node[Normal Node,draw=rwth-red] at (\vertex) {};
  \foreach \vertex in {s3,t3}
    \node[Normal Node,draw=rwth-purple] at (\vertex) {};
  \node[Label,anchor=south,inner sep=1pt] at ($(s1)+(-0.35,-0.1)$) {$s_1$};
  \node[Label,anchor=west,inner sep=1pt] at (6.05,2.20) {$t_1$};
  \node[Label,anchor=north,inner sep=1pt] at ($(s2)+(-0.35,0.05)$) {$s_2$};
  \node[Label,anchor=south,inner sep=1pt] at ($(t2)+(0.35,-0.1)$) {$t_2$};
  \node[Label,anchor=north,inner sep=1pt] at ($(s3)+(0.35,0.05)$) {$s_3$};
  \node[Label,anchor=west,inner sep=1pt] at (6.05,.95) {$t_3$};
}
  \fTwoCoordinates
  \fTwoBaseGraph
  \fTwoSampledPaths
  \fTwoOriginalNodes
\end{tikzpicture}%
            }
            \caption{Sampled paths $\mathcal P$.}
            \label{fig:1.93-b}
        \end{subfigure}

        \par\medskip

        \begin{subfigure}[t]{0.48\linewidth}
            \centering
            \resizebox{\linewidth}{!}{
                \begin{tikzpicture}[reset cm, yscale=0.65, xscale=0.9]
\def\fTwoCoordinates{%
  \coordinate (a) at (.50,2.20);
  \coordinate (b) at (1.80,2.20);
  \coordinate (c) at (3.10,2.20);
  \coordinate (d) at (4.45,2.20);
  \coordinate (e) at (.50,.95);
  \coordinate (f) at (1.80,.95);
  \coordinate (g) at (3.10,.95);
  \coordinate (h) at (4.45,.95);
  \coordinate (i) at (3.10,-.15);
  \coordinate (s1) at (.75,3.35);
  \coordinate (t1) at (5.85,2.20);
  \coordinate (s2) at (1.80,-.15);
  \coordinate (t2) at (3.10,3.35);
  \coordinate (s3) at (4.45,-.15);
  \coordinate (t3) at (5.85,.95);
}
\def\fTwoBaseGraph{%
  \foreach \from/\to in {
    a/b,a/e,e/f,f/g,b/c,c/d,c/g,c/t2,e/s2,g/s2,g/i,i/s2,
    s1/b,d/t1,f/s2,b/f,b/t2,h/s3,h/t3,i/s3,g/h}
    \draw[Normal Edge] (\from)--(\to);
}
\def\fTwoSampledPaths{%
  \draw[Normal Edge,draw=rwth-green, densely dotted] (s1)--(b)--(c)--(d)--(t1);
  \draw[Removable1,dashed] (s2)--(f)--(b)--(t2);
  \draw[Normal Edge,draw=rwth-purple,dashdotted] (s3)--(h)--(t3);
}
\def\fTwoLiftedTours{%
  \draw[Normal Edge, dashdotted, draw=rwth-turquoise] (b)--(a)--(e)--(s2);
  \draw[Normal Edge, dashdotted, draw=rwth-turquoise] (s2)--(g)--(i)--cycle;
}
\def\fTwoOriginalNodes{%
  \foreach \vertex in {a,e,i}
    \node[Normal Node] at (\vertex) {};
}
  \fTwoCoordinates
  \fTwoOriginalNodes
  \draw[UnusedEdge] (g)--(t3);
  \draw[UnusedEdge] (i)--(t3);
  \draw[Normal Edge, dashdotted, draw=rwth-turquoise] (b)--(a)--(e)--cycle;
  \draw[Normal Edge, dashdotted, draw=rwth-turquoise] (b)--(g)--(i)--cycle;
  \node[Depot, draw=white] at (s1) {};
  \node[Label,anchor=south,inner sep=1pt] at ($(s1)+(-0.35,-0.1)$) {\phantom{$s_1$}};
  \node[Label,anchor=west,inner sep=1pt] at (6.05,2.20) {\phantom{$t_1$}};
  \node[Label,anchor=north,inner sep=1pt] at ($(s2)+(-0.35,0.05)$) {\phantom{$s_2$}};
  \node[Label,anchor=south,inner sep=1pt] at ($(t2)+(0.35,-0.1)$) {\phantom{$t_2$}};
  \node[Label,anchor=west, inner sep=1pt] at (6.02,.95) {\phantom{$t_3$}};
  \node[Depot] at (b) {};
  \node[Depot] at (t3) {};
  \foreach \vertex in {a,e,g,i}
    \node[Normal Node] at (\vertex) {};
  \node[Label,anchor=south,inner sep=1pt] at (1.80,2.4) {$d_{12}$};
  \node[Label,anchor=north,inner sep=1pt] at (5.85,.75) {$d_3$};
\end{tikzpicture}%
            }
            \caption{Contracted Graphic Multi-Depot TSP solution.}
            \label{fig:1.93-c}
        \end{subfigure}
        \hfill
        \begin{subfigure}[t]{0.48\linewidth}
            \centering
            \resizebox{\linewidth}{!}{
                \begin{tikzpicture}[reset cm, yscale=0.65, xscale=0.9]
\def\fTwoCoordinates{%
  \coordinate (a) at (.50,2.20);
  \coordinate (b) at (1.80,2.20);
  \coordinate (c) at (3.10,2.20);
  \coordinate (d) at (4.45,2.20);
  \coordinate (e) at (.50,.95);
  \coordinate (f) at (1.80,.95);
  \coordinate (g) at (3.10,.95);
  \coordinate (h) at (4.45,.95);
  \coordinate (i) at (3.10,-.15);
  \coordinate (s1) at (.75,3.35);
  \coordinate (t1) at (5.85,2.20);
  \coordinate (s2) at (1.80,-.15);
  \coordinate (t2) at (3.10,3.35);
  \coordinate (s3) at (4.45,-.15);
  \coordinate (t3) at (5.85,.95);
}
\def\fTwoBaseGraph{%
  \foreach \from/\to in {
    e/f,f/g,c/g,c/t2,i/s3,g/h}
    \draw[UnusedEdge] (\from)--(\to);
}
\def\fTwoSampledPaths{%
  \draw[Normal Edge,draw=rwth-green, dashed] (s1)--(b)--(c)--(d)--(t1);
  \draw[Removable1,dashed] (s2)--(f)--(b)--(t2);
  \draw[Normal Edge,draw=rwth-purple,dashed] (s3)--(h)--(t3);
}
\def\fTwoLiftedTours{%
  \draw[Normal Edge, dashdotted, draw=rwth-turquoise] (b)--(a)--(e)--(s2);
  \draw[Normal Edge, dashdotted, draw=rwth-turquoise] (s2)--(g)--(i)--cycle;
}
\def\fTwoOriginalNodes{%
  \foreach \vertex in {a,c,d,e,f,g,h,i}
    \node[Normal Node] at (\vertex) {};
  \foreach \vertex in {s1,t1}
    \node[Normal Node,draw=rwth-green] at (\vertex) {};
  \foreach \vertex in {t2}
    \node[Normal Node,draw=rwth-red] at (\vertex) {};
  \foreach \vertex in {s3,t3}
    \node[Normal Node,draw=rwth-purple] at (\vertex) {};
  \foreach \vertex in {b}
    \node[Gray] at (\vertex) {};
  \foreach \vertex in {s2}
    \node[Gray,,draw=rwth-red] at (\vertex) {};
  \node[Label,anchor=south,inner sep=1pt] at ($(s1)+(-0.35,-0.1)$) {$s_1$};
  \node[Label,anchor=west,inner sep=1pt] at (6.05,2.20) {$t_1$};
  \node[Label,anchor=north,inner sep=1pt] at ($(s2)+(-0.35,0.05)$) {$s_2$};
  \node[Label,anchor=south,inner sep=1pt] at ($(t2)+(0.35,-0.1)$) {$t_2$};
  \node[Label,anchor=north,inner sep=1pt] at ($(s3)+(0.35,0.05)$) {$s_3$};
  \node[Label,anchor=west,inner sep=1pt] at (6.05,.95) {$t_3$};
}
  \fTwoCoordinates
  \fTwoBaseGraph
  \fTwoSampledPaths
  \fTwoLiftedTours
  \draw[Join, densely dotted] (b) to[bend left=22] (f);
  \draw[Join, densely dotted] (f) to[bend left=22] (s2);
  \fTwoOriginalNodes
  \node[Label,text=rwth-blue,fill=white,inner sep=1pt] at (2.17,1.57) {$J$};
\end{tikzpicture}%
            }
            \caption{Expanded solution $F_2$.}
            \label{fig:1.93-d}
        \end{subfigure}
    \end{minipage}

    \caption{Construction of the second candidate solution $F_2$. In (b), we sample paths indicated by colored and dashed edges. In (c), the components induced by the sampled paths are contracted into square depots and a \Problem solution is computed and shown in turquoise, dash-dotted edges. In (d), the expanded solution with wrong-parity vertices in gray is repaired by adding the edge copies $J$ shown in dotted blue. The computed solution is the union of all non-solid edges in (d).}
    \label{fig:1.93}
\end{figure}

We expand the contracted vertices $D'$ in the solution of $I'$, that is, after this step we have the current partial solution $\ALG(I') \, \dot{\cup} \, \P$. This set might not be a feasible Graphic Multi-Path TSP solution, since vertices $v \in V(\P)$ might have an odd degree $\deg_{\ALG(I')}(v)$ when only considering edges in $\ALG(I')$. 
Denote the set of these wrong-degree-parity vertices by $W$. In a final parity correction step we compute a minimum cardinality $W$-join $J$ in $\P$, and then output the solution $F_2 \coloneqq \P \, \dot{\cup} \, \ALG(I') \, \dot{\cup} \, J$. A visualization of this construction is given in \cref{fig:1.93}.
We upper bound the cost of this parity correction step using the following Lemma.

\begin{restatable}{lemma}{lemmaftwocostjoin}
\label{lem:f2_cost_join}
    Let $\P \, \dot{\cup} \, \ALG(I')$ be the partial solution after solving the auxiliary \Problem instance. We can compute a set $J \subseteq \P$ in polynomial time such that
    \[
        |J| 
        \leq |V(\P) \setminus (S\cup T)| + 2\sum_{v \in S} \ell_v
    \]
    and $F_2 \coloneqq \P \, \dot{\cup} \, \ALG(I') \, \dot{\cup} \, J$ is a feasible solution for the Graphic Multi-Path TSP.
\end{restatable}
The proof of this lemma first shows that $|J| \leq  |V(\P)| - |D'|$. When expanding the contracted components, the only vertices with odd degree are vertices in $V(\P)$ and every component has an even number of such defects. Parity correction can then be done with a join contained in a spanning tree of that component. Summing over the different components yields $|V(\P)| - |D'|$. The inequality $|V(\P)| - |D'| \leq |V(\P) \setminus (S\cup T)| + 2\sum_{v \in S} \ell_v$ is a straightforward calculation. Feasibility follows as each component of $F_2$ contains at least one source-sink pair $(s_i, t_i)$ and parity is correct by construction. We defer the complete proof to Appendix \ref{proof:cost_parity_f2}.

By summing up the costs of the sampling step, the solution of the auxiliary instance and the parity correction step we get the following bound for the expected cost of $F_2$.

\begin{lemma}
\label{lem:cost_F2}
    The solution $F_2$ has an expected cost of $\mathbb{E}[|F_2|] \leq \tfrac{3}{2}|\OPT\,| + \Phi$.
\end{lemma}
\begin{proof}
    We have proved feasibility of $F_2$ in \cref{lem:f2_cost_join}. To bound $|F_2| = |\P| + |\ALG(I')| + |J|$ we use \cref{eq:bounds_lv} for $|\P|$, \cref{thm:gmdm-tsp-general} and \cref{obs:obj_value_aux_instance} for $|\ALG(I')|$ and \cref{lem:f2_cost_join} for $|J|$ and obtain
    \begin{equation}
    \label{eq:bound_F2}
    |F_2| \leq \sum_{v\in V}\ell_v + \frac{3}{2}|\OPT| + \left(|V(\P) \setminus (S\cup T)| + 2 \sum_{v \in S } \ell_v \right).
    \end{equation}
    Taking expectations in \eqref{eq:bound_F2} and using Equations \eqref{eq:bounds_lv} and \eqref{eq:bounds_qv}  yields
    \[\mathbb{E}[|F_2|] \leq \frac{3}{2}|\OPT| + \sum_{v\not\in S\cup T} (z_v^P + 1 - q_v) + 3\sum_{v\in S\cup T} z_v^P = \frac{3}{2}|\OPT| + \Phi. \hfill\qedhere\]
\end{proof}

\subsection{Combining the Two Candidate Solutions}\label{sec:combining_cases_new}

Our algorithm returns the cheaper one of the two candidate solutions $F_1$ and $F_2$ constructed in Sections \ref{sec:first_candidate} and \ref{sec:second_candidate}, respectively. The proof of the approximation ratio is a straightforward combination of Lemmas \ref{lem:cost_F1} and \ref{lem:cost_F2}.

\begin{proof}[Proof of \cref{thm:multipath_improved}]
\cref{alg:multi_path} returns a solution of expected cost 
\[\mathbb{E}[|F|] \leq \min\{\mathbb{E}[|F_1|], \mathbb{E}[|F_2|]\} \leq \min\left\{2|\OPT| - \frac{e-2}{2e-1}\Phi, \ \frac{3}{2}|\OPT| + \Phi\right\}.\]
To upper bound this minimum in terms of $|\OPT|$, we let $x \coloneqq \Phi/|\OPT|$ and compute the intersection of the two lines
\[f_1(x) = 2 - \frac{e-2}{2e-1}x \quad\text{and}\quad f_2(x) = \frac{3}{2} + x.\]
These lines intersect at $x = \tfrac{2e-1}{6e-6}$, and plugging in this value for $\Phi$ yields
\[\min\{\mathbb{E}[|F_1|], \mathbb{E}[|F_2|]\} \leq \left(\frac{3}{2}+\frac{2e-1}{6e-6}\right)|\OPT| = \frac{11e-10}{6e-6}|\OPT| < 1.9304\cdot |\OPT|.\hfill\qedhere\]
\end{proof}

\section*{Acknowledgment of Usage of AI}
We proved \cref{thm:gmdm-tsp-general} without any use of AI. For the setting of subcubic graphs considered in \cref{sec:gmdmtsp_subcubic}, we presented Claude Fable 5 with a preliminary version of \cref{thm:gmdm-tsp-subcubic}. This version was weaker in two ways: We only considered subcubic graphs where depots had a degree of at most two, and our approximation guarantee included an additional additive term depending on the structure of an optimal solution. Claude Fable 5 found and fixed mistakes in our proof and managed to avoid the degree constraint for depot vertices. Subsequently, with ChatGPT 5.6 Sol we got rid of the additive term in the approximation guarantee. For \GraphicMultiPath considered in \cref{sec:graphic_multpath}, we obtained a 1.9522-approximation without using AI. We presented ChatGPT 5.6 Sol with this result and the idea of the algorithm presented in \cref{sec:graphic_multpath}. Given these, the model found an improved analysis yielding \cref{thm:multipath_improved}.

\bibliography{bibliography}

@InProceedings{2010_1_5_approx_mdmtsp,
    author="Xu, Zhou and Rodrigues, Brian",
    editor="Kaplan, Haim",
    title="A 3/2-Approximation Algorithm for Multiple Depot Multiple Traveling Salesman Problem",
    booktitle="Algorithm Theory - SWAT 2010",
    year="2010",
    publisher="Springer Berlin Heidelberg",
    address="Berlin, Heidelberg",
    pages="127--138",
    isbn="978-3-642-13731-0"
}

@article{2017_xu_rodrigues,
    title = {An extension of the {C}hristofides heuristic for the generalized multiple depot multiple traveling salesmen problem},
    journal = {European Journal of Operational Research},
    volume = {257},
    number = {3},
    pages = {735-745},
    year = {2017},
    issn = {0377-2217},
    doi = {https://doi.org/10.1016/j.ejor.2016.08.054},
    url = {https://www.sciencedirect.com/science/article/pii/S0377221716306932},
    author = {Zhou Xu and Brian Rodrigues},
}

@InProceedings{graphic_multipath_paper,
    author = {Morteza Alimi and Niklas Dahlmeier and Tobias Mömke and Philipp Pabst and Laura Vargas Koch},
    title = {Approximating Graphic Multi-Path {TSP} and Graphic Ordered {TSP}},
    booktitle = {2026 SIAM Symposium on Simplicity in Algorithms (SOSA)},
    chapter = {},
    pages = {524-537},
    doi = {10.1137/1.9781611978964.43},
    URL = {https://epubs.siam.org/doi/pdf/10.1137/1.9781611978964.43},
    eprint = {2509.00448},
    publisher = {Society for Industrial and Applied Mathematics},
    year = {2026},
}

@InProceedings{deppertKaulMnich_MDMTSP,
  author =	{Deppert, Max and Kaul, Matthias and Mnich, Matthias},
  title =	{A ($3/2 + \varepsilon$)-Approximation for Multiple {TSP} with a Variable Number of Depots},
  booktitle =	{31st Annual European Symposium on Algorithms (ESA 2023)},
  pages =	{39:1--39:15},
  series =	{Leibniz International Proceedings in Informatics (LIPIcs)},
  ISBN =	{978-3-95977-295-2},
  ISSN =	{1868-8969},
  year =	{2023},
  volume =	{274},
  editor =	{G{\o}rtz, Inge Li and Farach-Colton, Martin and Puglisi, Simon J. and Herman, Grzegorz},
  publisher =	{Schloss Dagstuhl -- Leibniz-Zentrum f{\"u}r Informatik},
  address =	{Dagstuhl, Germany},
  URL =		{https://drops.dagstuhl.de/entities/document/10.4230/LIPIcs.ESA.2023.39},
  URN =		{urn:nbn:de:0030-drops-186925},
  doi =		{10.4230/LIPIcs.ESA.2023.39}
}

@article{mömke_svensson,
author = {M{\"o}mke, Tobias and Svensson, Ola},
title = {Removing and Adding Edges for the Traveling Salesman Problem},
year = {2016},
issue_date = {March 2016},
publisher = {Association for Computing Machinery},
address = {New York, NY, USA},
volume = {63},
number = {1},
issn = {0004-5411},
url = {https://doi.org/10.1145/2739008},
doi = {10.1145/2739008},
journal = {Journal of the ACM},
month = feb,
articleno = {2},
numpages = {28}
}

@inproceedings{böhm2026approximatingtravelingsalesmanproblems,
    author = {Martin Böhm and Zachary Friggstad and Tobias Mömke and Joachim Spoerhase},
    title = {Approximating Traveling Salesman Problems Using a Bridge Lemma},
    booktitle = {Proceedings of the 2025 Annual ACM-SIAM Symposium on Discrete Algorithms (SODA)},
    chapter = {},
    pages = {1166-1177},
    doi = {10.1137/1.9781611978322.34},
    URL = {https://epubs.siam.org/doi/abs/10.1137/1.9781611978322.34},
    year = {2025},
}

@article{seboVygen2014graphic,
author = {Seb{\H o}, Andr{\'a}s and Vygen, Jens},
title = {Shorter tours by nicer ears: 7/5-Approximation for the graph-{TSP}, 3/2 for the path version, and 4/3 for two-edge-connected subgraphs},
year = {2014},
issue_date = {October 2014},
publisher = {Springer-Verlag},
address = {Berlin, Heidelberg},
volume = {34},
number = {5},
issn = {0209-9683},
url = {https://doi.org/10.1007/s00493-014-2960-3},
doi = {10.1007/s00493-014-2960-3},
journal = {Combinatorica},
month = oct,
pages = {597–629},
numpages = {33}
}

@article{XU2011218,
title = {An analysis of the extended {C}hristofides heuristic for the $k$-depot {TSP}},
journal = {Operations Research Letters},
volume = {39},
number = {3},
pages = {218-223},
year = {2011},
issn = {0167-6377},
doi = {https://doi.org/10.1016/j.orl.2011.03.002},
url = {https://www.sciencedirect.com/science/article/pii/S0167637711000290},
author = {Zhou Xu and Liang Xu and Brian Rodrigues}
}

@book{Traub_Vygen_2024,
place={Cambridge},
title={Approximation Algorithms for Traveling Salesman Problems},
publisher={Cambridge University Press},
author={Traub, Vera and Vygen, Jens},
year={2024}
}

@inproceedings{OveisGharanSaberiSingh2011GraphicTSP,
author = {Gharan, Shayan Oveis and Saberi, Amin and Singh, Mohit},
title = {A Randomized Rounding Approach to the Traveling Salesman Problem},
year = {2011},
isbn = {9780769545714},
publisher = {IEEE Computer Society},
address = {USA},
url = {https://doi.org/10.1109/FOCS.2011.80},
doi = {10.1109/FOCS.2011.80},
booktitle = {Proceedings of the 2011 IEEE 52nd Annual Symposium on Foundations of Computer Science},
pages = {550–559},
numpages = {10},
series = {FOCS '11}
}

@InProceedings{mucha2012graphicTSP,
  author =	{Mucha, Marcin},
  title =	{{13/9-approximation for Graphic TSP}},
  booktitle =	{29th International Symposium on Theoretical Aspects of Computer Science (STACS 2012)},
  pages =	{30--41},
  series =	{Leibniz International Proceedings in Informatics (LIPIcs)},
  ISBN =	{978-3-939897-35-4},
  ISSN =	{1868-8969},
  year =	{2012},
  volume =	{14},
  editor =	{D\"{u}rr, Christoph and Wilke, Thomas},
  publisher =	{Schloss Dagstuhl -- Leibniz-Zentrum f{\"u}r Informatik},
  address =	{Dagstuhl, Germany},
  URL =		{https://drops.dagstuhl.de/entities/document/10.4230/LIPIcs.STACS.2012.30},
  URN =		{urn:nbn:de:0030-drops-34025},
  doi =		{10.4230/LIPIcs.STACS.2012.30}
}

@InProceedings{dvorakKralMohar2017cubic,
  author =	{Zden{\v{e}}k Dvo{\v{r}}{\'a}k and Daniel Kr{\'a}l and Bojan Mohar},
  title =	{{Graphic TSP in Cubic Graphs}},
  booktitle =	{34th Symposium on Theoretical Aspects of Computer Science (STACS 2017)},
  pages =	{27:1--27:13},
  series =	{Leibniz International Proceedings in Informatics (LIPIcs)},
  ISBN =	{978-3-95977-028-6},
  ISSN =	{1868-8969},
  year =	{2017},
  volume =	{66},
  editor =	{Heribert Vollmer and Brigitte Vall{\'e}e},
  publisher =	{Schloss Dagstuhl--Leibniz-Zentrum fuer Informatik},
  address =	{Dagstuhl, Germany},
  URL =		{http://drops.dagstuhl.de/opus/volltexte/2017/7006},
  URN =		{urn:nbn:de:0030-drops-70068},
  doi =		{10.4230/LIPIcs.STACS.2017.27}
}

@article{WigalYooYu2021cubic,
title = {Approximating {TSP} walks in subcubic graphs},
journal = {Journal of Combinatorial Theory, Series B},
volume = {158},
pages = {70-104},
year = {2023},
issn = {0095-8956},
doi = {https://doi.org/10.1016/j.jctb.2022.09.002},
url = {https://www.sciencedirect.com/science/article/pii/S0095895622000879},
author = {Michael C. Wigal and Youngho Yoo and Xingxing Yu},
}

@article{serdyukov,
  author   = {Serdyukov, A. I.},
  title    = {Some extremal bypasses in graphs},
  journal  = {Upravlyaemye Sistemy},
  volume   = {17},
  pages    = {76--79},
  year     = {1978},
  language = {russian},
  note     = {In Russian}
}

@techreport{christofides,
  author      = {Christofides, Nicos},
  title       = {Worst-case analysis of a new heuristic for the traveling salesman problem},
  institution = {Graduate School of Industrial Administration, Carnegie-Mellon University},
  address     = {Pittsburgh},
  type        = {Technical Report},
  number      = {388},
  year        = {1976},
  note        = {Reprinted in \emph{Operations Research Forum} 3 (2022), Article 20,
                 \url{https://doi.org/10.1007/s43069-021-00101-z}}
}

@article{an_kleinberg_shmoys_s_t_path_best_of_many_christofides,
    author = {An, Hyung-Chan and Kleinberg, Robert and Shmoys, David B.},
    title = {Improving {C}hristofides' Algorithm for the $s$-$t$ Path {TSP}},
    year = {2015},
    issue_date = {November 2015},
    publisher = {Association for Computing Machinery},
    address = {New York, NY, USA},
    volume = {62},
    number = {5},
    issn = {0004-5411},
    url = {https://doi.org/10.1145/2818310},
    doi = {10.1145/2818310},
    journal = {Journal of the ACM},
    month = nov,
    articleno = {34},
    numpages = {28}
}

@article{traub_t_tours_best_of_many_christofides,
    title = {Improving on {Best-of-Many-Christofides} for ${T}$-tours},
    journal = {Operations Research Letters},
    volume = {48},
    number = {6},
    pages = {798-804},
    year = {2020},
    issn = {0167-6377},
    doi = {https://doi.org/10.1016/j.orl.2020.09.009},
    url = {https://www.sciencedirect.com/science/article/pii/S0167637720301474},
    author = {Vera Traub},
}

@misc{one_plus_square_root_paper,
      title={A $(1+1/\sqrt{2})$-Approximation for the Multiple-Depot Traveling Salesman Problem}, 
      author={Jingyang Zhao and Yuxi Liu and Mingyu Xiao},
      year={2026},
      eprint={2609.19537},
      archivePrefix={arXiv},
      primaryClass={cs.DS},
      url={https://arxiv.org/abs/2609.19537}, 
}

@misc{zhaoXuSiao2026betterThanThreeHalf,
      title={Improved Approximation Algorithms for the Multiple-Depot Split Delivery Vehicle Routing Problem}, 
      author={Jingyang Zhao and Yonghang Su and Mingyu Xiao},
      year={2026},
      eprint={2601.01841},
      archivePrefix={arXiv},
      primaryClass={cs.DS},
      url={https://arxiv.org/abs/2601.01841}, 
}

@inproceedings{reducing_path_tsp_to_tsp,
    author = {Traub, Vera and Vygen, Jens and Zenklusen, Rico},
    title = {Reducing path {TSP} to {TSP}},
    year = {2020},
    isbn = {9781450369794},
    publisher = {Association for Computing Machinery},
    address = {New York, NY, USA},
    url = {https://doi.org/10.1145/3357713.3384256},
    doi = {10.1145/3357713.3384256},
    booktitle = {Proceedings of the 52nd Annual ACM SIGACT Symposium on Theory of Computing},
    pages = {14–27},
    numpages = {14},
    location = {Chicago, IL, USA},
    series = {STOC 2020}
}

\newpage
\appendix

\section{\texorpdfstring{Proofs from \cref{sec:gmdmtsp_subcubic}}{Proofs from Section 3}}\label{sec:appendix_sec_3}
\existenceofcheapjoins*
\begin{proof}
    For $H = (V,E)$ we consider the following $T$-join polytope which is the convex hull of the incidence vectors of $\odd(H)$-joins by Theorem 2.21 from \cite{Traub_Vygen_2024}.
    \begin{align*}
  \lvert F\rvert - x(F) + x(\delta(U)\setminus F) &\ge  1
    &&(U \subseteq V,\; F \subseteq \delta(U), (| U \cap \odd(H)| + |F| \odd )\\
  x_e &\le 1 &&(e \in E) \\
  x_e &\ge 0 &&(e \in E). 
\end{align*}
    Let $x \in [0,1]^E$ with $x(e) = \frac{1}{3}$ for all $e \in E$. The first inequality reduces to $\frac{1}{3} |F| + \frac{1}{3}|\delta(U)| \geq 1$. By Lemma 2.20 in \cite{Traub_Vygen_2024}, we obtain $|\delta(U)| \equiv |U \cap \odd(H)| \mod 2$ and hence that $|\delta(U)| + |F|$ is odd. Combining this with $|\delta(U)| \geq 2$, since $H$ is $2$-edge-connected, we conclude $|\delta(U)| + |F| \geq 3$. Thus, $x$ is in our polytope, and for a min cost $\odd(H)$-join $J$ we have $c(J) \leq \frac{1}{3}c(E(H))$.
\end{proof}

\section{\texorpdfstring{Proofs from \Cref{sec:graphic_multpath}}{Proofs from Section 4}}\label{sec:appendix_sec_4}

Before proving the missing lemmas from \cref{sec:graphic_multpath}, we begin by giving a pseudocode description of our algorithm and the simplified version of the LP used in \cite{graphic_multipath_paper}.

\begin{algorithm}[h]
\caption{1.9304-approximation for Graphic Multi-Path TSP}\label{alg:multi_path}
\textbf{Input} Graph $G = (V,E)$, source-sink pairs $(s_i, t_i)_{i \in [k]}$
\begin{algorithmic}[1]
\vspace*{0.3cm}

\Statex \noindent\textbf{Sampling Phase:}
\State{Compute an optimal fractional solution $(x,z)$ to (LP).}
\State{Decompose each $x_i$ into paths $P_i^j$ and cycles $C_i^j$, discard the cycles.}
\State{For each $i\in [k]$ sample one path $P_i$ from the decomposition with probabilities equal to the coefficients $\lambda_i^j$.}
\State{By $\P$ denote the multiset that is the union of all $P_i$.}
\vspace*{0.3cm}

\Statex \noindent\textbf{Construction of $F_1$:}
\State{Greedily reconnect each vertex not covered by $\P$ for a cost of 2 per vertex.}

\vspace*{0.3cm}
\Statex \noindent\textbf{Construction of $F_2$:}
\State{Construct an auxiliary \Problem-instance $I' = (G', D')$ by contracting the components of $\mathcal{P}$ into a single vertex each, and letting $D'$ be the set of these contracted vertices.}
\State{Compute a $3/2$-approximation for $I'$ using \cref{thm:gmdm-tsp-general}.}
\State{Expand the set of contracted vertices $D'$ in the computed solution back to the union of paths $\mathcal{P}$ to obtain $\tilde{F_2}$.}
\State{Perform the parity correction described in \cref{lem:f2_cost_join} to obtain $F_2$.}
\vspace*{0.3cm}

\noindent\Return{the cheaper solution $F$ among $F_1$ and $F_2$.}
\end{algorithmic}
\end{algorithm}

$\begin{alignedat}{3} \label{LP}
 & \text{minimize} & \sum_{v \in V} z_v & \\
 & \text{s.t.} \ & x_i \text{ is a $s_i$-$t_i$-flow of value 1} &\qquad\forall i\in[k], \\
 && z_v = \sum_{i\in[k]} x_i(\delta^{+}(v)) & \qquad \forall v \in V, \\
 && z_{v} \geq 1 & \qquad \forall v \in V \setminus T.
\end{alignedat}$

\lemmacostfone*
\begin{proof}
\label{proof:cost_f1}
    Throughout the analysis we write $\gamma \coloneqq \tfrac{e-2}{2e-1} \approx 0.1619$. By definition of $\gamma$ it holds that
    \begin{equation}
        \label{eq:prop_gamma}
        1-2\gamma - (2-\gamma)e^{-1} = 0.
    \end{equation}
    The expected cost of $|F_1|$ is given by
    \[\mathbb{E}[|F_1|] \leq \sum_{v\in V} z_v^P + 2\sum_{v\not\in S\cup T} q_v,\]
    with the first summand accounting for the cost of the sampling step and the second summand accounting for the cost of the reconnection step. 
    
    It suffices to show that each vertex contributes at most $2z_v - \gamma\varphi_v$ to the right hand side of the inequality above. The statement then follows by summing over all $v \in V$. 
    
    First, consider a terminal vertex $v \in S \cup T$. Using $3\gamma < 1$, the cost of $v$ for the algorithm is at most
    \[z_v^P \leq 2z_v - z_v^P \leq 2z_v - 3\gamma z_v^P = 2z_v - \gamma\varphi_v.\]

    Now consider the non-terminal vertices $v \not\in S\cup T$. They incur a cost of $z_v^P + 2q_v$ for the algorithm, and thus we need to show that $z_v^P + 2q_v \leq 2z_v - \gamma\varphi_v$, or equivalently
    \[2z_v - z_v^P - 2q_v \geq \gamma(z_v^P + (1 - q_v)).\]
    Note that $q_v$ has coefficient $-2$ on the left hand side and $-\gamma$ on the right hand side. Since $-(2-\gamma)<0$, we may assume that $q_v$ is maximized.

The following bound on $q_v$ (the probability that $v$ is not covered by a path) is shown in \cite{graphic_multipath_paper}. It also follows directly as start and end vertices are by definition part of a path. Roughly speaking, for all other vertices, their total path flow value $z_v^P$ distributes to paths with some fraction $\lambda_j$ and the probability of not being covered is then $\prod (1-\lambda_j) \leq e^{-z_v^P}$. Thus, we have
   $q_v \leq e^{-z_v^P}$, and can now assume  $q_v = e^{-z_v^P}$. \\

    \noindent\emph{Case 1: $z_v^P \leq 1$.} We use $z_v \geq 1$ for all non-terminal vertices. Thus, it suffices to show that $g \geq 0$ where
    \[g(z_v^P) \coloneqq (2-\gamma) - (1+\gamma)z_v^P - (2-\gamma)e^{-z_v^P}.\]
    This is a concave function on $[0,1]$, so we only have to verify $g(0) \geq 0$ and $g(1) \geq 0$. Clearly, we have $g(0) = 0$, and
    \[g(1) = 1 - 2\gamma - (2-\gamma)e^{-1} = 0\]
    by \cref{eq:prop_gamma}. \\

    \noindent\emph{Case 2: $z_v^P \geq 1$.} We use $z_v \geq z_v^P$, and hence, it suffices to show that $h \geq 0$ where
    \[h(z_v^P) = (1-\gamma)z_v^P - \gamma - (2-\gamma)e^{-z_v^P}.\]
    The derivative of this function fulfills $h'(z_v^P) = (1-\gamma) + (2-\gamma)e^{-z_v^P} > 0$ for all $z_v^P \geq 1$. Since $h(1) = g(1) = 0$ we conclude $h \geq 0$ which finishes the proof.
\end{proof}

\lemmaftwocostjoin*

\begin{proof}
\label{proof:cost_parity_f2}
We split the proof into three parts. First, we show $|J| \leq  |V(\P)| - |D'| = |V(\P)| - \operatorname{comp}(V(\P), \P)$, subsequently we prove that $|V(\P)| - \operatorname{comp}(V(\P), \P) \leq |V(\P) \setminus (S\cup T)| + 2\sum_{v \in S} \ell_v$ and lastly we show feasibility of the solution.

Recall that by $W$ we denote the set of vertices that have an odd degree in $\ALG(I')$. As $\ALG(I')$ is Eulerian in $G'$ and expanding the contracted vertices in $I'$ does not change the degree of vertices that are not part of $\mathcal{P}$, we must have $W \subseteq V(\P)$. Furthermore, let $C$ be an arbitrary component of $\P$. $|W \cap V(C)|$ must be even since the vertex $v_C$ that $C$ got contracted into in $I'$ has even degree in $\ALG(I')$.

Fix some arbitrary spanning tree $T_C$ of $C$, and let $J_C \subseteq T_C$ be the unique $(W \cap V(C))$-join in $T_C$. As $J_C$ is a subset of $T_C$, we must have $|J_C|\leq |V(C)|-1$. Setting $J = \bigcup_C J_C$ it follows that
\[
    |J| \leq \sum_C (|V(C)|-1) = |V(\P)| - \operatorname{comp}(V(\P), \P).
\]

\begin{claim}
\label{lem:f2_cost_terminals}
    It holds that $|V(\P)|  - \operatorname{comp}(V(\P), \P) \leq |V(\P) \setminus (S\cup T)| + 2\sum_{v \in S} \ell_v$.
\end{claim}
\begin{claimproof}
    Let $v \in V(\P)$ be some isolated terminal vertex in $\P$. Observe that we do not have to pay for the vertex $v$ on the left hand side since it forms a component on its own. Let $\overline{V(\P)}$ denote the set of vertices in $V(\P)$ that are part of some non-empty sampled path. Then,
\begin{align*}
    |V(\mathcal P)| - \operatorname{comp}(V(\mathcal P),\mathcal P)
    &\le |\overline{V(\mathcal P)}| \\
    &\le |\overline{V(\mathcal P)} \setminus (S \cup T)|
       + |\overline{V(\mathcal P)} \cap (S \cup T)|.
\end{align*}
Clearly, the number of vertices
$v \in \overline{V(\mathcal P)} \cap S$ with $\ell_v>0$
is bounded by $\sum_{v\in S}\ell_v$.
For each vertex $t \in \overline{V(\mathcal P)} \cap (S \cup T)$
not counted above, it must be the final vertex of some non-empty
sampled path, of which there are at most $\sum_{v\in S}\ell_v$,
since each contributes an edge leaving its source.
Indeed, either $t\in T\setminus S$, so its own commodity path
is non-empty, or $t\in S$ and $\ell_t=0$, so any non-empty
sampled path containing $t$ ends there.
Thus, $|\overline{V(\mathcal P)}\cap(S\cup T)|$ is bounded by
$2\sum_{v\in S}\ell_v$. This finishes the proof.
\end{claimproof}
By construction, every vertex has even degree in $\ALG(I') \disjcup J$. It remains to verify that $F_2 = \P \, \dot{\cup} \, \ALG(I') \, \dot{\cup} \, J$ is a feasible solution.
For this, iterate over all connected components $C$ of $\ALG(I') \, \dot{\cup} \, J$. $C$ has even degree at every vertex and is connected and thus has an Euler tour. Additionally, $C$ must include at least one vertex from $\mathcal{P}$ since the contracted components of $\P$ are exactly the depots of $I'$. Choose some sampled path $P_i$ for which $C$ intersects $P_i$, and attach the Euler tour of $C$ to $P_i$. The result is a new $s_i$-$t_i$ walk $P_i'$. After iterating through all components $C$, this results in a collection of $s_i$-$t_i$-walks whose multiset union is exactly $F_2$. Additionally, $F_2$ covers all of $V$: The vertices in $\P$ are covered by the original sampled paths, and the vertices in $V \setminus V(\P)$ are covered by $\ALG(I')$.
\end{proof}

\section{Derandomization of the Graphic Multi-Path TSP Algorithm}
\label{app:derandomization}

In this section we describe how \cref{alg:multi_path} can be derandomized using the method of conditional expectations. Notice that (given a fixed optimal LP-solution $(x,z)$) the variable $\Phi$ that we use to compare the two solutions is fully deterministic, i.e., it does not depend on the choice of paths $\P$. Thus, we can derandomize the construction of the two candidate solutions $F_1$ and $F_2$ independently of each other. To derandomize $F_1$, we use the exact approach from \cite{graphic_multipath_paper} and inductively choose a path $P_i$ for each commodity $i \in [k]$ such that
\[\mathbb{E}[|F_1| \ | \ (P_1, \ldots, P_i) \text{ are fixed}] \leq 2\cdot \OPT_{\text{LP}} - \frac{e-2}{2e-1}\Phi\]
which deterministically yields a solution satisfying the cost bound from \cref{lem:cost_F1}.

For $F_2$ we cannot directly use the same approach, since $\mathbb{E}[|F_2|]$ depends on $\ALG(I')$ and thus, the exact expected value conditioned on a set of chosen paths $(P_1, \ldots, P_i)$ is much more difficult to compute. However, it is sufficient to inductively choose paths $P_1, \ldots, P_k$ which maintain the cost bound
\[\mathbb{E}[|F_2| \ | \ (P_1, \ldots, P_i) \text{ are fixed}] \leq \frac{3}{2}|\OPT| + \Phi\]
from \cref{lem:cost_F2} for all $i \in [k]$. Assume that $(P_1, \ldots, P_{i-1})$ are already chosen. Recall that in \cref{lem:cost_F2} we uniformly bound $|\ALG(I')|$ by $\tfrac{3}{2}|\OPT|$ independently of $\P$. Thus, using \cref{eq:bound_F2} it suffices to choose paths $P_i$ for which
\[\mathbb{E}\left[\sum_{v \in V} \ell_v + |V(\P) \setminus (S \cup T)| + 2\sum_{v \in S}\ell_v \ \Big| \ (P_1, \ldots, P_i) \text{ are fixed}\right] \leq \Phi.\]
By induction, this bound holds when choosing $P_i$ at random, and thus for the path $P_i$ minimizing the left hand side the inequality must hold. Furthermore, this minimizing path can be found efficiently since for all $v \in V$ both $\mathbb{E}[\ell_v \ | \ (P_1, \ldots, P_i) \text{ are fixed}]$ and $\Pr[v \text{ is sampled} \ | \ (P_1, \ldots, P_i) \text{ are fixed}]$ can be computed easily. After choosing all paths $(P_1, \ldots, P_k)$, the bound from \cref{lem:cost_F2} still holds, and thus the construction of $F_2$ can be derandomized.

\end{document}